\documentclass[a4paper,USenglish,cleveref,autoref,thm-restate]{lipics-v2021}
\usepackage[dvipsnames]{xcolor}
\usepackage[most]{tcolorbox}
\usepackage{lipsum}  
\usepackage[utf8]{inputenc}
\usepackage{tcolorbox}
\usepackage{amssymb,amsmath,amsthm}
\usepackage{xspace}
\usepackage{todonotes}
\usepackage{colortbl}

\usepackage{tcolorbox}
\tcbuselibrary{skins} 
\usepackage{xcolor}

\nolinenumbers

\usepackage{mathtools}

\usepackage[T1]{fontenc}
\definecolor{anti-flashwhite}{rgb}{0.95, 0.95, 0.96}

\newcommand{\dla}{{\sc Dilation $(2+\epsilon)$-Augmentation}\xspace}

\definecolor{tyred}{rgb}{0.8, 0.0, 0.0}
\colorlet{mix}{red!50!black}

\providecommand{\keywords}[1]
{
	\textbf{\textit{Keywords---}} #1
}
\usepackage{hyperref}
\hypersetup{
	pdfnewwindow=true,      
	colorlinks=true,       
	linkcolor=mix,          
	citecolor=blue,        
	urlcolor=blue          
}

\newcommand{\V}{\mathcal{V}}
\newcommand{\I}{\mathcal{I}}

\newcommand{\dta}{{\sc Dilation $t$-Augmentation}\xspace}
\newcommand{\dtwoa}{{\sc Dilation $2$-Augmentation}\xspace}
\newcommand{\dtai}[1]{{\sc Dilation $#1$-Augmentation}\xspace}
\newcommand{\order}[1]{\mathcal{O}(#1)}

\newcommand{\yes}{\textup{\textsc{yes}}\xspace}
\newcommand{\no}{\textup{\textsc{no}}\xspace}
\newtheorem{red_rule}{Reduction Rule}
\newcommand{\fpt}{{\sf FPT}\xspace}
\newcommand{\polytime}{\textup{\textsf{P}}\xspace}

\newcommand{\nph}{{\sf NP}-hard\xspace}
\newcommand{\woh}{{\sf W[1]}-hard\xspace}
\newcommand{\wth}{{\sf W[2]}-hard\xspace}

\newcommand{\kdd}{\mathcal{K}_{d,d}}
\newcommand{\CG}{\mathbb{C}\xspace}
\newcommand{\rr}{Reduction Rule}

\newcommand{\lr}[1]{\left(#1\right)}
\newcommand{\LR}[1]{\left\{#1\right\}}
\newcommand{\cO}{\mathcal{O}}

\usepackage{algorithm}
\usepackage{algorithmicx}
\usepackage{algpseudocode}
\usepackage{mdframed}

\algtext*{EndWhile}
\algtext*{EndIf}
\algtext*{EndFor}

\title{\LARGE Multivariate Exploration of Metric Dilation}
\titlerunning{ Multivariate Exploration of Metric Dilation}

\author{Aritra Banik}{National Institute of Science, Education and Research, An OCC of Homi Bhabha National Institute, Bhubaneswar 752050,
Odisha, India}{aritra@niser.ac.in}{}{}

\author{Fedor V. Fomin}{University of Bergen, Norway}{Fedor.Fomin@uib.no}{https://orcid.org/0000-0003-1955-4612}{Supported by the Research Council of Norway via the project  BWCA (grant no. 314528).}

\author{Petr A. Golovach}{University of Bergen, Norway}{Petr.Golovach@ii.uib.no}{https://orcid.org/0000-0002-2619-2990}{Supported by the Research Council of Norway via the project  BWCA (grant no. 314528).}

\author{Tanmay Inamdar}
{Indian Institute of Technology Jodhpur, Jodhpur, India}{taninamdar@gmail.com}{https://orcid.org/0000-0002-0184-5932}{Supported by  IITJ Research Initiation Grant (grant number I/RIG/TNI/20240072).}

\author{Satyabrata Jana}{University of Warwick, UK}{satyamtma@gmail.com}{https://orcid.org/0000-0002-7046-0091}{Supported by the Engineering and Physical Sciences Research Council (EPSRC)
via the project MULTIPROCESS (grant no. EP/V044621/1).}

\author{Saket Saurabh}{The Institute of Mathematical Sciences, HBNI, Chennai, India  \and University of Bergen, Norway }{saket@imsc.res.in}{https://orcid.org/0000-0001-7847-6402}{The author is supported by the European Research Council (ERC) under the European Union's Horizon 2020 research and innovation programme (grant agreement No. 819416); and he also acknowledges the support of Swarnajayanti Fellowship grant DST/SJF/MSA-01/2017-18.}

\Copyright{Aritra Banik, Fedor V. Fomin, Petr A. Golovach, Tanmay Inamdar, Satyabrata Jana, Saket Saurabh}

\authorrunning{A. Banik, F. V. Fomin, P. A. Golovach, T. Inamdar, S. Jana, S. Saurabh}


\EventEditors{John Q. Open and Joan R. Access}
\EventNoEds{2}
\EventLongTitle{42nd Conference on Very Important Topics (CVIT 2016)}
\EventShortTitle{CVIT 2016}
\EventAcronym{CVIT}
\EventYear{2016}
\EventDate{December 24--27, 2016}
\EventLocation{Little Whinging, United Kingdom}
\EventLogo{}
\SeriesVolume{42}
\ArticleNo{23}

\begin{document}

\maketitle

\begin{abstract}
  Let $G$ be a weighted graph embedded in a metric space $(M, d_M)$. The vertices of $G$ correspond to the points in $M$, with the weight of each edge $uv$ being the distance $d_M(u,v)$ between their respective points in $M$. The dilation (or stretch) of $G$ is defined as the minimum factor $t$ such that, for any pair of vertices $u,v$, the distance between $u$ and $v$—represented by the weight of a shortest $u,v$-path—is at most $t\cdot d_M(u,v)$. We study {\sc Dilation $t$-Augmentation}, where the objective is, given a metric $M$, a graph $G$, and numerical values $k$ and $t$, to determine whether $G$ can be transformed into a graph with dilation $t$ by adding at most $k$ edges.

Our primary focus is on the scenario where the metric $M$ is the shortest path metric of an unweighted graph $\Gamma$. Even in this specific case, {\sc Dilation $t$-Augmentation} remains computationally challenging. In particular, the problem is {\sf W[2]}-hard parameterized by $k$ when $\Gamma$ is a complete graph, already for $t=2$. Our main contribution lies in providing new insights into the impact of combinations of various parameters on the computational complexity of the problem. We establish the following. 
\begin{itemize}
\item The parameterized dichotomy of the problem with respect to dilation $t$, when the graph $G$ is sparse: Parameterized by $k$, the problem is {\sf FPT} for graphs excluding a biclique $K_{d,d}$ as a subgraph for $t\leq 2$ and the problem is {\sf W[1]}-hard for $t\geq 3$ even if $G$ is a forest consisting of disjoint stars.   
\item The problem is {\sf FPT} parameterized by the combined parameter $k+t+\Delta$, where $\Delta$ is the maximum degree of the graph $G$ or $\Gamma$.   
\end{itemize}

\end{abstract}

\section{Introduction} \label{Sec:intro}
Consider a finite metric space $\mathcal{M}=(V, d_M)$, and let $G$ be a sparse weighted graph with vertices corresponding to the points in $V$. The weights assigned to the edges of $G$ represent the distances in $\mathcal{M}$ between their end-points. That is, the graph
$G = (V,E)$ is embedded in a metric space $\mathcal{M}=(V, d_M)$. The graph $G$ is called a \emph{$t$-spanner} if, for every pair of vertices $u,v \in V$, the distance between them in $G$ is at most $t\cdot d_M(u,v)$. The concept of spanners, introduced by Peleg and Sch{\"{a}}ffer \cite{PelegS89}, has evolved into a fundamental tool in various domains, including algorithms, distributed computing, networking, data structures and metric geometry, as highlighted in \cite{narasimhan2007geometric}.  The minimum number $t$ for which $G$ is a $t$-spanner defines the \emph{stretch} or \emph{dilation} of $G$.

In their book~\cite[p.474, Problem 9]{narasimhan2007geometric}, Narasimhan and Smid presented the problem of enhancing the dilation of a spanner by adding at most $k$ edges: Develop an efficient algorithm to identify the $k > 1$ edges that minimize (or approximately minimize) the stretch factor of the resulting geometric graph. Formally, the problem is defined as follows. 

	\begin{tcolorbox}[enhanced,title={\color{black} \sc{Dilation $t$-Augmentation}}, colback=white, boxrule=0.4pt,
		attach boxed title to top center={xshift=-3.5cm, yshift*=-2.5mm},
		boxed title style={size=small,frame hidden,colback=white}]
		
		\textbf{Input:} \hspace*{5mm} A graph $G = (V,E)$ embedded in a metric space $\mathcal{M}=(V, d_M)$ and \hspace*{19mm} an integer  $k$.\\
		\textbf{Question:} Does there exist a set of $k$ edges $S \subseteq V \times V$ such that the  dilation of  \hspace*{19mm} $G' = (V,E\cup S)$  is at most $t$?
	\end{tcolorbox}


The case where $k=1$, was investigated by Farshi, Giannopoulos, Gudmundsson~\cite{FarshiGG08}, Luo,  Wulff-Nilsen \cite{LuoW08}, and Wulff-Nilsen~\cite{Wulff-Nilsen10}.
However, for the more general scenario where $k>1$, the problem becomes significantly more challenging and poorly understood. Giannopoulos et al. \cite{GiannopoulosKKKM10} and Gudmundsson and Smid~\cite{GudmundssonS09} demonstrated that obtaining the best dilation spanner by adding $k$ edges to an empty graph is already \nph. Gudmundsson and Wong \cite{GudmundssonW22} proposed an algorithm that in $\cO(n^3 \log n)$ time, identifies $k$ edges whose addition to $G$ results in a graph with a stretch factor within $\cO(k)$ of the minimum stretch factor.
Related problems have also been studied in geometric settings, e.g., Aronov et al.~\cite{AronovBCGHSV08}, who showed  that, given a set $S$ of $n$ points in $\mathbb{R}^2$, and an integer $0 \leq k < n$, we can construct a  {\em geometric graph} with vertex set $S$,  at most $n-1 +k$ edges, maximum degree five, and dilation $\cO(n / (k + 1))$  in time $\cO(n \log n)$. 


We approach the {\sc{Dilation $t$-Augmentation}} problem from the perspective of Multivariate (Parameterized) Complexity -- another popular paradigm to deal with intractable problems~\cite{DBLP:books/sp/CyganFKLMPPS15,DBLP:series/mcs/DowneyF99,DBLP:series/txtcs/FlumG06}. The two most natural parameters associated with the problem are the size of the solution $k$ (the size of the augmented set) and the stretch factor $t$. We are looking for an algorithm with running time $f(k)\cdot n^{\cO(1)}$ or $g(k,t) \cdot n^{\cO(1)}$. Here, $n=|V(G)|$. These algorithms are called fixed-parameter tractable (\fpt) algorithms. There is also an accompanying theory of {\sf W}-hardness that allows us to show that problems do not admit \fpt~algorithms (see~\cite{DBLP:books/sp/CyganFKLMPPS15,DBLP:series/mcs/DowneyF99,DBLP:series/txtcs/FlumG06} for further details). The problem of finding a $t$-spanner for a given graph is considered from the perspective of parameterized complexity in ~\cite{Kobayashi18a,Kobayashi20}.
Observe that, \textsc{Dilation $t$-Augmentation} admits an algorithm with running time $n^{\cO(k)}$ -- try all possible subsets of size $k$ of $V\times V$ as $S$. Thus, this naturally leads to the following question. 

\begin{mdframed}[backgroundcolor=gray!10,topline=false,bottomline=false,leftline=false,rightline=false] 
 \centering
 Does {\rm \textsc{Dilation $t$-Augmentation}} admit an  \fpt~algorithm?
\end{mdframed}


A simple result shows that, in its full generality, the problem is \wth. In fact, consider $\mathcal{M}=(V, d_M)$ derived from an unweighted clique $K$ on $n$ vertices. That is, the vertices of the metric correspond to the vertices of $K$, and $d_M(u,v)$ is the length of shortest path between $u$ and $v$ in $K$, which is $1$. Then, for $t=2$, \textsc{Dilation $2$-Augmentation} corresponds to adding edges to $G$ so that the diameter of the augmented graph becomes $2$ (see Figure \ref{fig:example} for an illustration). This is the well-known {\sc Diameter $2$-Augmentation} problem which is known to be \wth~\cite{DBLP:journals/dam/GaoHN13}. Thus, we do not expect that \textsc{Dilation $t$-Augmentation} to admit an algorithm with running time $g(k)\cdot n^{\cO(t)}$  (See Proposition~\ref{prop:diamaug} in the appendix for completeness). This simple hardness result motivates the following set of questions: 

\begin{sloppypar}
\begin{itemize}
    \item For which metrics $\mathcal{M}$, does \textsc{Dilation $t$-Augmentation} admit an \fpt~algorithm?
    \item For which family of input graphs $\cal G$, does \textsc{Dilation $t$-Augmentation} admit an \fpt~algorithm?
    \item For which pairs of family of input graphs and metric, $({\cal G}, \mathcal{M})$, does \textsc{Dilation $t$-Augmentation} admit an \fpt~algorithm?
\end{itemize}
\end{sloppypar}

In this article, we specifically concentrate on the fundamental and non-trivial scenario where the metric $\mathcal{M}$ corresponds to the shortest-path metric of an unweighted graph. Our contribution represents a substantial addition to the limited body of literature that addresses the challenging problem of \dta.

\subsection{Our Model, Results, and Methods}
\begin{figure}[t!]
	\centering
	\includegraphics[scale=0.7]{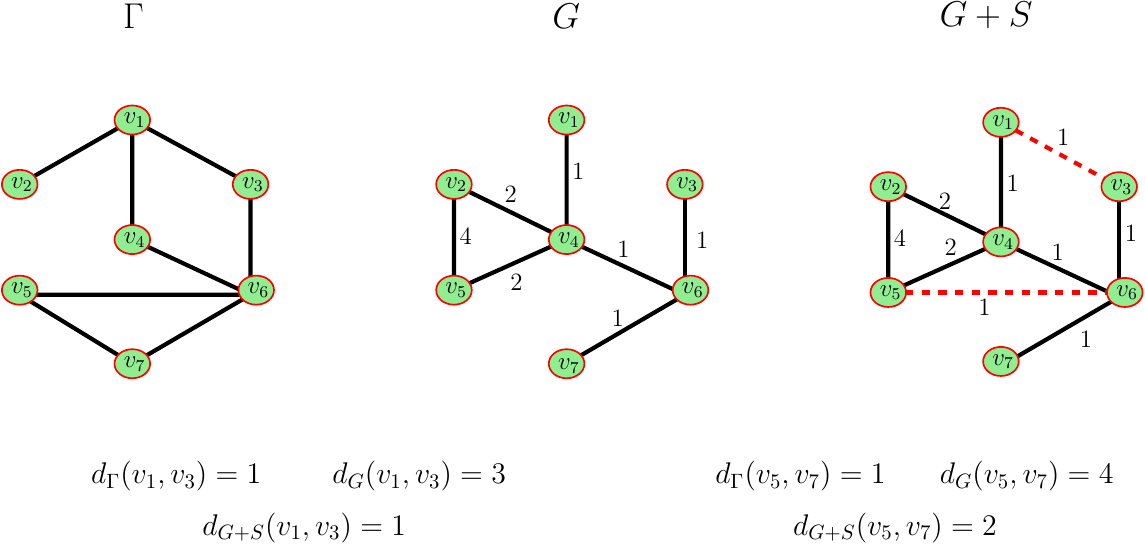}
	\caption{An instance of \dtwoa with $k=2$.  The edges of the solution $S$ are shown in dashed red. The edge weights in $G$ are derived from the corresponding shortest path in $\Gamma$.  }

 \label{fig:example}
\end{figure}
Throughout this paper, we deal with the shortest-path metric derived from an unweighted graph, unless otherwise mentioned.\footnote{Any metric can be derived from an edge weighted graph. Thus, our assumption represents a simplification.} For an edge weighted graph $H$, let $d_{H}(u,v)$ denote the {\em weighted shortest path distance} in the graph $H$ between two vertices $u$ and $v$. Additionally, we use ${\sf hop}_H(u,v)$ to denote the length of the shortest path between $u$ and $v$ when we forget the edge weights. In other words, this denotes the length of the shortest path in $H$ when all edges receive weight $1$. We will use $d_H$ and ${\sf hop}_H$ {\em only when $H$ is edge weighted}; otherwise, both $d_H$ and ${\sf hop}_H$ represent the same thing. 
Throughout the article, the metric $\cal M$ will be derived from an undirected unweighted graph $\Gamma=(V,E)$. That is, the vertices of the metric correspond to the vertices of $\Gamma$, and $d_M(u,v)$ is assigned the shortest path distance in $\Gamma$, between $u$ and $v$. Observe that $d_M(u,v)=d_\Gamma(u,v)={\sf hop}_\Gamma(u,v)$. We will use $(G, \Gamma, k)$ as an instance of \textsc{Dilation $t$-Augmentation} problem. 
Since $\cal M$ is derived from $\Gamma$, for ease of notation, we use $d_\Gamma(u,v)$ instead of $d_M(u,v)$. 
Furthermore, recall that $G$ is embedded in the metric space $\cal M$ derived from $\Gamma$. That is, $G$ is an edge-weighted graph, where the weights assigned to the edges of $G$ represent the distances in $\Gamma$ between their end-points. That is, $w(u,v)=d_\Gamma(u,v)$ and $d_G(u,v)$ denote the {\em weighted shortest path distance} in $G$.

The starting point of our research is the result of Giannopoulos et al.~\cite{GiannopoulosKKKM10} and Gudmundsson and Smid~\cite{GudmundssonS09} who showed 
that finding the best dilation spanner by adding $k$ edges to an empty graph is \nph. 
It also follows from the results of Peleg and Sch{\"{a}}ffer~\cite{PelegS89} (see Proposition~\ref{prop:empty}) that \dta is \nph for $t=2$, where $G$ is an empty graph and the metric $\cal M$ is derived from an undirected graph $\Gamma$.  So, a natural question is whether this special subcase is \fpt. Starting with this question, we trace the boundaries of the \dta problem by instantiating different graph families to which $\Gamma$ can belong or by 
 instantiating different graph families to which $G$ can belong. Our results consist of the following.

  \begin{mdframed}[backgroundcolor=cyan!5,topline=false,bottomline=false,leftline=false,rightline=false] 
 \begin{enumerate}

 \item Our main algorithmic result is an \fpt algorithm for \dtwoa when $G$ belongs to the family of $\kdd$-free graphs. Recall that a graph $G$ is $\kdd$-free if it does not contain a complete bipartite graph with $d$ vertices, each on both sides of the bipartition, as a subgraph. However, there is no restriction on $\Gamma$.

 We complement this result by showing that this result cannot be extended for $t=3$. Indeed, we show that when $G$ is a disjoint union of a star and an independent set, and $\Gamma$ is an arbitrary graph, then \dtai{3} is \woh. In a slight converse we also show that \dtai{3} is \wth when $\Gamma$ is a star, and $G$ is an arbitrary graph. On the other hand, in a generalization of the latter setting when $\Gamma$ is a tree, we show that \dtai{2} is polynomial-time solvable. 

  \item Observe that for the families of stars we cannot bound the maximum degree of each graph uniformly. We show that if $G$ (resp.~$\Gamma$) belongs to the family of graphs with a maximum degree at most $d$ and $\Gamma$ (resp.~$G$) is an arbitrary graph, then \dta admits an algorithm \fpt with running time $f(k,t,d) \cdot n^{\cO(1)}$.

 We complement this result by showing that this result cannot be extended for the weighted metric. That is, when $G$ belongs to the family of graphs of maximum degree $3$ and the graph $\Gamma$ is edge-weighted, then the problem becomes intractable.  In fact, the edges of $\Gamma$ get weights from the two-size set $\{1,w\}$. We show that \dtai{(2+\epsilon)} is \woh in this case. 
 \end{enumerate}
 \end{mdframed}

 It is important to note that the family of $\kdd$-free graphs includes trees, planar graphs, graphs that exclude a fixed graph $H$ as a minor, graphs of bounded expansion, nowhere dense graphs, and graphs of bounded degeneracy.\footnote{For a formal definition of many of these graph classes, we refer to standard texts on graph theory, e.g., Diestel~\cite{diestel2024graph}.}

We refer to Table~\ref{table:introresults} for a quick reference of all the results shown in this article.

\begin{table}[h]
 \begin{center}
		\begingroup
		\setlength{\tabcolsep}{7pt} 
		\renewcommand{\arraystretch}{1.5} 
		\begin{tabular}{ c c c c c}
			\hline
   \rowcolor{anti-flashwhite}
   
			  $\boldsymbol{G}$ &  \textbf{Metric} ($\boldsymbol{\Gamma}$) & \textbf{Parameter(s)} & $\boldsymbol{t}$ & \textbf{Complexity} \\\hline\hline
	$\kdd$-free & General	&  $k+d$ & 2 & \fpt~(Section~\ref{sec:kddFree})\\	
    Star forest & General	&  $k $ & 3 & \woh~(Section~\ref{sec:genstar3})\\
    General & Star	  & $k $ & 3 & \woh~(Section~\ref{sec:stargen3})\\
    General &  Tree & $k $ & 2 & \polytime~(Section~\ref{sec:treegen2})\\\hline
    General   & Bounded degree	&  $k+t+\Delta$ & $t$ & \fpt~(Section~\ref{sec:boundgen}) \\
	Bounded degree & General &  $k+t+\Delta$ & $t$ & \fpt~(Section~\ref{sec:genbound}) \\
    Subcubic & $(1,w)$-weighted &  $k $ & $2+\epsilon$ & \woh~(Section~\ref{Sec:hard_weight})\\\hline
    Edgeless & General	  & $k $ & 2 & \nph~(Section \ref{sec:genempty})\\	General & Clique &  $k$ & 2 & \wth~(Section \ref{sec:cliquegen}) \\\hline
\end{tabular}
	\caption{Complexity of the {\sc Dilation $t$-Augmentation} for different $G$, metric, parameters and $t$.\label{table:introresults}}
		\endgroup
  \end{center}
  \end{table}
 
Our algorithmic results are based on the structural interaction between $\Gamma$ and $G$. We start by defining the notion of conflicting pairs (a pair $(u,v)$ in $G$ such that $d_G(u,v)>t \cdot d_\Gamma(u,v)$) and show that to resolve these pairs it suffices to focus on those pairs that are adjacent in $\Gamma$ (that is, $(u,v)\in E(\Gamma)$).   We call them adjacent conflicts. Resolving adjacent conflict pairs is the key to our results, both algorithms and hardness. 

Our main result regarding \dtwoa when $G$ is $\kdd$-free is given in Section~\ref{sec:kddFree}. Our algorithm can be broadly divided into three separate phases. First, in Section~\ref{sec:cgvc} we define the notion of a \emph{conflict graph} and bound the size of the minimum vertex cover in this graph using the restrictions imposed on the solution for \dta in the special case of $t = 2$. Then, in Section~\ref{sec:boundvc}, we analyze the interplay between the bounded-size vertex cover, along with the structural properties of $G$ due to the fact that it is $\kdd$-free, and use it to design a recursive algorithm that tries to guess a certain kind of edges that must belong to any solution in a \yes-instance. At the leaves of this recursive procedure, we can upper bound the total number of vertices involved in conflicts. Nevertheless, we cannot get rid of all other vertices since they may be crucial for certain shortest paths. Finally, in Section~\ref{sec:solveinstance} we bound the size of the instance by carefully eliminating a subset of conflict-free vertices, at which point we can enumerate all possible solutions.
For the other two algorithmic results, we obtain a small set of $V$ that contains all the end-points of edges of the solution. The hardness results are shown via parameter preserving reductions from known {\sf W}-hard problems, using classical gadget constructions.






         


\subparagraph{Notations.} Let $[n]$ be the set of integers $\{1,\ldots, n\}$. 
    For any graph $G$, we denote the set of vertices of $G$ by $V(G)$ and the set of edges by $E(G)$. We denote the set of non-edges by $\overline{E}(G)$, that is, $\overline{E}(G)= {V(G)\choose 2} \setminus E(G)$, where $\binom{V(G)}{2}$ is the set of all unordered pairs of distinct vertices in $V(G)$. For notational convenience, even though our edges/non-edges are undirected, we use the notation $(a, b)$ instead of $\LR{a, b}$ -- note that due to this convention, $(a, b) = (b, a)$. A graph $G = (V, E)$ is said to be a \emph{star}, if there exists a vertex $c \in V(G)$ such that $E = \LR{ (u, c): u \in V(G) \setminus \LR{c} }$. In this case, we say that $c$ is the \emph{center} of the star. Vertices of $V(G) \setminus \LR{c}$ (if any) are said to be the \emph{leaves} of the star. If every connected component of a graph $G$ is a star, then we say that $G$ is a \emph{star forest}.
    
    For a graph $G=(V,E)$ and a subset $S$ of edges in $\overline{E}(G)$, 
   we use the notation $G+S$ to mean the graph $G'=(V,E\cup S)$. 
    For a graph $G$, a path is a sequence of distinct vertices $v_1 \cdots v_\ell$ such that $(v_i,v_{i + 1})\in E(G)$ for $i \in[\ell-1]$ and this path is denoted by $\langle v_1, \cdots, v_\ell \rangle$.     For an undirected unweighted graph $H$, we use ${\sf hop}_H(u,v)$ to denote the length of the shortest path between $u$ and $v$. For a graph $H$, a vertex $v$ and an integer $\ell$, $N_H^\ell(v)$ denotes all vertices $w$ such that ${\sf hop}_H(v,w)\leq \ell$. Similarly, for a graph $H$, a vertex subset $W$ and an integer $\ell$, $N_H^\ell(W)=\cup_{w\in W} N_H^\ell(w)$.


\section{Conflict versus Adjacent Conflicts}
Let $(G,\Gamma,k)$ be an instance of \textsc{Dilation $t$-Augmentation} problem.  We say that two vertices $u$ and $v$ in $G$ are in  {\em $t$-conflict} if
$d_G(u,v)>t \cdot d_\Gamma(u,v)$. Furthermore, two vertices are said to be in {\em adjacent $t$-conflict} with each other, if $u$ and $v$ are in conflict and $u$ and $v$ are adjacent in $\Gamma$. We say that a graph $G$ is {\em   (adjacent) $t$-conflict-free}  if there is no pair of vertices $(u,v)$ such that $u$ and $v$ are in (adjacent) $t$-conflict. Throughout the discussion, the value of $t$ will be clear from the context, so we use the terms \emph{conflict/conflict-free} instead of \emph{$t$-conflict/$t$-conflict-free}. The following lemma shows the relationship between conflict-free and adjacent conflict-free. 

    \begin{lemma}\label{lem:adjacent}
        For any set of edges $S$, $G+S$ is conflict-free if and only if $G+S$ is adjacent conflict-free.
    \end{lemma}
    \begin{proof}
        If $G+S$ is conflict-free, then $G+S$ is also adjacent conflict-free by definition. For the reverse direction, let $G+S$ be adjacent conflict-free. Let $u$ and $v$ be two distinct vertices in $V$. We will show that $d_{G+S}(u,v) \leq t \cdot d_\Gamma(u,v)$. 
        Towards this, consider a shortest path between $u$ and $v$ in $\Gamma$. Let $\langle u=v_1,v_2,\cdots ,v_{\lambda}=v \rangle$ be one such shortest path between $u$ and $v$ in $\Gamma$. Observe that since each $(v_i,v_{i+1})$ is adjacent in $\Gamma$ and $G+S$ is adjacent to conflict-free, $d_{G+S}(v_i,v_{i+1})\leq t \cdot  d_{\Gamma}(v_i,v_{i+1})$. Now $d_{G+S}(u,v) \leq \sum_{i=1}^{\lambda-1}d_{G+S}(v_i, v_{i+1}) \leq \sum_{i=1}^{\lambda-1}t \cdot  d_{\Gamma}(v_i,v_{i+1}) = t \cdot d_{\Gamma}(u,v)$.          Thus, there exists a path between $u$ and $v$ in $G+S$ of length at most $t \cdot d_{\Gamma}(u,v)$. This completes the proof.
    \end{proof}


  \begin{mdframed}[backgroundcolor=cyan!5,topline=false,bottomline=false,leftline=false,rightline=false] 
  For our problem, we add edges to $G$ to make it conflict-free. Lemma~\ref{lem:adjacent} shows that for this purpose it is sufficient to focus {\em only} on adjacent conflicts. Thus, from now on, we only focus on adjacent conflicts. Unless explicitly mentioned otherwise, by \emph{conflict}, we mean adjacent conflict. Nevertheless, we may use \emph{adjacent conflict} at some places to emphasize the adjacent-ness of the conflict.
  
\end{mdframed}

The next result formalizes the fact that the distances in $G$ are lower bounded by the distances in $\Gamma$. The proof follows from definition, and hence omitted. 

    \begin{observation}\label{obs:path}
    Let $G = (V,E)$ be a graph embedded in a metric space derived from the undirected graph $\Gamma$. Then for any pair of vertices $x$ and $y$ we have $d_{\Gamma}(x,y) \leq d_G(x,y)$.
    \end{observation}


  \begin{mdframed}[backgroundcolor=cyan!5,topline=false,bottomline=false,leftline=false,rightline=false] 
 Let $(G,\Gamma,k)$ be a \yes-instance of \dta and $S$ be an (unknown) minimal solution. The set $S$ is also called dilation $t$-augmentation set.  We use $V_S$ to denote the set of end-points of the edges in $S$. Let $V_c$ denote the set of all vertices in $G$ that are in conflict with some other vertex. In the next two sections, $V_c$ and $V_S$ will be used repeatedly. 
\end{mdframed}

\section{Dilation $2$-Augmentation for $\kdd$-free Graphs}\label{sec:kddFree}
Let $(G,\Gamma,k)$ be an instance of \dtwoa.  In this section, we consider the case where $G$ belongs to the family of $\kdd$-free graphs. Recall that a graph $G$ is $\kdd$-free if it does not contain a complete bipartite graph with $d$ vertices each on both sides of the bipartition. However, there is no restriction on $\Gamma$. 


\subsection{Conflict Graph and Vertex Cover} \label{sec:cgvc}
We first define a {\em conflict graph} $\CG$ on the set of vertices $V(G)$, which captures all adjacent conflicts.  We place an edge between two vertices $u$ and $v$ in $\CG$ if and only if $u$ and $v$ are in adjacent conflict with each other.  Note that $V_c$ (the set of vertices present in adjacent conflicts) is exactly the set of vertices with degree at least one in $\CG$. The next result shows that every edge in $E(\CG)$ intersects some edge of a dilation $2$-augmentation set $S$.

\begin{lemma}\label{lemma:oneendpoint}
 Let $(G,\Gamma,k)$ be a \yes-instance of \dtwoa and let $S$ be a minimal solution to the given instance. In addition, let $\CG$ be the corresponding conflict graph. Then, for any edge $(u,v)\in E(\CG)$, $|\{u,v\}\cap V_S|\geq 1$.  

\end{lemma}

\begin{proof}
For the sake of contradiction, assume that there exists an edge $(u,v) \in  E(\CG)$ such that $\LR{u, v} \cap V_S = \emptyset$. Recall that by the definition of an adjacent conflict $d_\Gamma(u,v)=1$. As $G+S$ is conflict-free, there is a (weighted) path of length $2$ of the form $\langle u,w,v \rangle$ between $u$ and $v$ in $G+S$. If, $(u,w)\in S$ or $(v,w)\in S$, then we have $|\{u,v\}\cap V_S|\geq 1$. So assume that neither $(u,w)\in S$ nor $(v,w)\in S$. Then, the path $\langle u,w,v \rangle$ of (weighted) length $2$ between $u$ and $v$ is in $G$, contradicting the fact that $u$ and $v$ are in adjacent conflict.   
\end{proof}

Lemma~\ref{lemma:oneendpoint} allows us to bound the size of a minimum vertex cover (set of vertices that intersects all the edges) of $\CG$. To do so, we propose the next reduction rule and prove its correctness (or safety). 

\begin{red_rule}\label{rr:lowvc}
    If the size of a maximum matching in $\CG$ is more than $2k$, return \no.
\end{red_rule}
\begin{lemma}
    \rr~\ref{rr:lowvc} is safe. 
\end{lemma}
\begin{proof}
 Let $(G,\Gamma,k)$ be a \yes-instance of \dtwoa and let $S$ be a minimal solution to the given instance. Note that the $k$ edges of $S$ can be incident to at most $2k$ vertices, in particular at most $2k$ vertices of $V_c$.  However, if $\CG$ has a matching $M$ of size larger than $2k$, then there exists at least one edge $e = (u, v)$ of $M$, such that $V_S \cap \left\{u, v\right\} = \emptyset$, contradicting Lemma~\ref{lemma:oneendpoint}. Thus, in a \yes-instance, the size of maximum matching in $\CG$ is bounded by $2k$. In other words, if the size of a maximum matching in $\CG$ is more than $2k$, we have a \no-instance. 
\end{proof}

First, we use a polynomial-time algorithm (e.g., \cite{edmonds1965paths}) to find the maximum matching $M$ in $\CG$ and apply \rr~\ref{rr:lowvc}.
Notice that if the reduction rule \ref{rr:lowvc} is not applicable, then $|M| \le 2k$, which implies that it has a vertex cover of size at most $4k$ -- in particular, we can take the set of end-points of the edges of in $M$ to be such a vertex cover, call it $R$. 


Let $I=V\setminus R$. Note that $I$ induces an independent set in $\CG$, but not necessarily in $G$. 
Next, we bound the degree of each vertex in $R$ in $\CG$.  This will imply that the number of conflict pairs is bounded. 

The set $R$ is at most of size $4k$. For our algorithm at this stage, we ``guess the edges of $S$ that have both end-points in $R$''. This results in an equivalent annotated instance, which is now formalized.  We first define an ``annotated'' instance of the problem.  An example of the annotated version of the problem is given by the tuple $(G', \Gamma, k', R')$, where we also provide a vertex cover $R'$ of the corresponding conflict graph, which is used mainly for book keeping purposes. The task is still to find a dilation $2$-augmentation set $S$ of size at most $k'$. Here, we are not allowed to add an edge to a solution that does not have an end-point outside $R$. 

Let $\overline{E}_{R}$ be the set of non-edges in $G[R]$. For every subset $\emptyset \subseteq E_j\subseteq \overline{E}_{R}$ of size at most $k$, we create an instance ${\cal I}_j$ of the annotated problem, as follows:  ${\cal I}_j \coloneqq (G_{j}, \Gamma, k_{j},R)$ where $G_{j}=G+E_j$ and $ k_j=k-|E_j|$. The following lemma is immediate. 

\begin{lemma}
    $(G, \Gamma, k)$ is a \yes-instance if and only if one of the instances in $\{(G_{j}, \Gamma, k_{j},R): \forall\  \emptyset \subseteq E_j \subseteq \overline{E}_R \text{ s.t. } k_j \le k \}$ is a \yes-instance. 
\end{lemma}

 \begin{mdframed}[backgroundcolor=cyan!5,topline=false,bottomline=false,leftline=false,rightline=false] 
 From now on, we assume that we have an annotated instance of \dtwoa. That is, 
 $(G, \Gamma, k,R)$ is an instance of \dtwoa and we are seeking a dilation $2$-augmentation set $S$ such that for any edge $(x,y)$ in $S$, $|\{x,y\}\cap R|\leq 1$. That is, every edge added must contain at least one end-point that does not belong to $R$.
\end{mdframed}

\subsection{Bounding the Size of the $V_c$ in Annotated Instances} \label{sec:boundvc}
We now describe how to solve the annotated version of the problem. As a first step, we give a recursive algorithm such that at the end we obtain instances such that the size of $V_c$ gets upper-bounded by a function of $k$ and $d$ alone. To this end, and to make it easier to understand, we adapt the following technical definition of \cite{BonnetBTW19} in the context of our problem.

\begin{sloppypar}
\begin{definition}[\cite{BonnetBTW19}]
    A decreasing \fpt-Turing-reduction is an \fpt algorithm which, given
an annotated instance $(G, \Gamma, k, R)$, produces $\ell = g(k, d)$ annotated instances $(G_1, \Gamma, k_1, R_1), \cdots ,(G_{\ell}, \Gamma, k_{\ell}, R_\ell)$, such that:
\begin{itemize}
    \item $(G, \Gamma, k, R)$ is a \yes-instance iff one of the annotated instances in $\{(G_{i},\Gamma, k_{j}, R_i):1\leq j\leq \ell\}$ is a \yes-instance.
    \item  $k_j\leq k-1$ for every $j\in[\ell]$.
\end{itemize}
\end{definition}
\end{sloppypar}

We now give a sequence of \fpt Turing reductions to solve the problem. At a high level, we will iterate over the vertices in $R$ that have high degrees in $\CG$, and at each step we will add at least one vertex from $I$ to $R$, and at least one new edge to $G$. Thus, in each step, the budget $k$ reduces by at least one. At least one of the resulting instances will correspond to a ``correct'' set of choices.




  

Note that initially the size of $R$ is at most $4k$; however during the sequence of Turing \fpt reductions, the size of $R$ may increase. Nevertheless we will maintain the invariant that $|R|$ remains bounded by $5k$. 

Let us define an auxiliary function $f$ that is useful for defining some parameters in the upcoming steps. This function is defined as follows.
\[f(i) = \begin{cases}
      d &\text{ if } i = d
      \\d \cdot k + k^2 + k &\text{ if } i = d-1
      \\d \cdot k^{d-i} +  k^{d-i+1} + \lr{2 \cdot \sum_{j = 2}^{d-i} k^j} + k &\text{ if } 0 \le i \le d-2
  \end{cases}\]
It is easy to verify that $f$ satisfies the following property.
\begin{proposition} \label{prop:function}
    For any $1 \le i \le d$, $f(i-1)=  (f(i) + k)\cdot k + k  $.
\end{proposition}

Thus, if each vertex of $R$ has at most $f(0)$ many neighbours in $I$ in $\CG$, then eventually $|V_c|$ gets bounded by $5k \cdot f(0)$. We then move to the final step described below after Corollary~\ref{cor:final}.

Now we are in the situation where there is a vertex in $R$ with  at least $f(0)+1$ neighbors in $I$ in $\CG$. Let $v$ be such a vertex in $R$.   Let $U = \LR{u_1, u_2, \ldots, u_{\kappa}} \subseteq I \cap N_{\CG}(v)$ be the set of neighbours of $v$ in $I$, where $\kappa > f(0)$.

\begin{red_rule}\label{rr:nohighdegree}
    If there is no vertex $w \in I $ such that $|U\cap N_G(w)| >  f(1)$  then we return \no. 
\end{red_rule}
\begin{lemma} \label{lem:nohighdegree}
    \rr~\ref{rr:nohighdegree} is safe. 
\end{lemma}
\begin{proof}
Assume that for every vertex $w \in I $, we have $|U\cap N_G(w)| \leq f(1)$. Then, we will show that we cannot resolve all conflicts involving $v$. Suppose $(G,\Gamma, k, R)$ is a \yes-instance of \dtwoa and $S$  be a minimial solution of size at most $k$. Since each $(v,u_i)$ is an adjacent conflict, we have $d_{G+S}(v,u_i)\leq 2$. This implies ${\sf hop}_{G+S}(v,u_i)\leq 2$. Therefore, the edge of $S$ that resolves the conflict $(v,u_i)$ must be adjacent to $v$ or $u_i$ or both. Fix an edge $(v,w)\in S$. Since $S$ is a solution to the annotated instance of the vertex $w\in I$, the number of paths of hop $2$ in $G+S$ starting at $v$, using $(v,w)$ and ending at a vertex in $U$ is upper-bounded by $|U\cap N_G(w)| + |S| \leq f(1) + k$. Therefore, the total number of conflicts that can be resolved by edges of the form $(v,w)$ is at most $(f(1)+k)|Z|$. Here, $Z\subseteq S$ contains edges of the form $(v,w)$.  Let us now focus on the conflicts that are resolved by edges of the form $(w,u_i)$, $u_i\in U$. This implies that we have paths of the form $\langle v,w,u_i\rangle$ such that $(v,w)\in E(G)$ and $(w,u_i)\in S$. The number of paths of this kind is upper bounded by $|S| - |Z|$. This implies that the total number of conflicts in which $v$ participates, and can be resolved by $S$, is bounded by $(f(1)+k)|Z| + |S| - |Z| \leq f(1)\cdot k+k^2+k= f(0)$. 
\end{proof}

Now, if \cref{rr:nohighdegree} is not applicable, then there must be a vertex $w \in I$ such that $w$ has at least $f(1)$ neighbors in $U$. We use $w$ to identify a small set of edges that must intersect every solution $S$ of size at most $k$.

\begin{lemma} \label{lem:kddset}
    There exists a polynomial-time algorithm to find a set $W_v = \LR{w_1, w_2, \ldots, w_\delta}$ such that (1) $\delta < d$, and (2) If we have a \yes-instance, then any solution $S$ of size at most $k$ satisfies that $S \cap \LR{ (v, w_i) : i \in [\delta] } \neq \emptyset$. 
   \end{lemma}

\begin{proof}
By our definition of $\CG$, the vertex $v$ is in adjacent conflict with each vertex in $U$. 
First, by \rr~\ref{rr:nohighdegree} and Lemma~\ref{lem:nohighdegree}, we know that for a \yes-instance there exists a vertex $w \in I$ such that $|U \cap N_G(w)| > d \cdot k^{d}$. We let this vertex to be $w_1$. Now either $(v, w_1) \in S$, and then we are done. Else we are in the case when $(v, w_1) \not\in S$. Let us define $W_1 \coloneqq \LR{w_1}$.

\begin{figure}[H]
            \centering
            \includegraphics[width=0.6\textwidth]{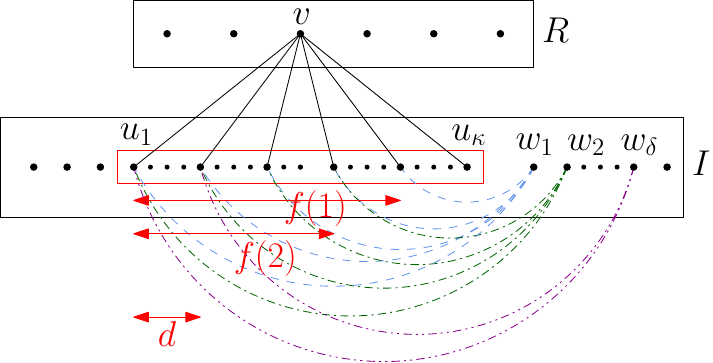}
             \caption{Illustration of arguments used in Lemma~\ref{lem:kddset}. }
\end{figure} \label{fig:kdd}

Suppose we have inductively found a set $W_{i-1} = \LR{w_1, \ldots, w_{i-1}}$ for some $i \ge 2$. Further, inductively assume that $U_{i-1} \coloneqq \bigcap_{j \in [i-1]} U \cap N_G(w_j)$ is such that $|U_{i-1}| > f(i-1)$. See Figure~\ref{fig:kdd} for an illustration. We prove the following crucial claim which lets us obtain the next vertex in the sequence.

\begin{claim} \label{cl:nextvertex}
    Suppose $(G,\Gamma, k, R)$ is a \yes-instance of \dtwoa and $S$  be a solution of size at most $k$. If $S \cap \LR{ (v, w_j) : j \in [i-1] }= \emptyset$. Then there exists a vertex $w_i \in I \setminus W_{i-1}$ such that
    $w_i$ has more than $f(i)$ neighbours in $G$, from $U_{i-1}$.
\end{claim} 
\begin{proof}
    Assume that for every vertex $w \in I \setminus W_{i-1} $, we have $|U_{i-1}\cap N_G(w)| \leq f(i)$. Then, we will show that we cannot resolve all conflicts involving $v$ with the vertices from $U_{i-1}$. Suppose $(G, \Gamma, k, R)$ is a \yes-instance of \dtwoa and $S$  is a solution of size at most $k$ such that $S \cap \LR{ (v, w_j) : j \in [i-1] }= \emptyset$. Now we have that for each $u \in U_{i-1}$, $(v,u)$ is an adjacent conflict but $d_{G+S}(v,u)\leq 2$. This implies ${\sf hop}_{G+S}(v,u)\leq 2$. Therefore, the edge of $S$ that resolves the conflict $(v,u)$ must be adjacent to $v$ or $u$ or both. Fix an edge $(v,w)\in S$. Since $S$ is a solution to the annotated instance of the vertex $w\in I$, and $w \notin W_{i-1}$ (by assumption),  the number of paths of hop $2$ in $G+S$ starting at $v$, using $(v,w)$ and ending at a vertex in $U_{i-1}$ is upper-bounded by $|U_{i-1}\cap N_G(w)| + |S| \le f(i) + k$. Therefore, the total number of conflicts $(v,u)$ where $u \in U_{i-1}$ that can be resolved by edges of the form $(v,w)$ where $w \notin W_{i-1}$ is at most $(f(i) + k) \cdot |Z| \le (f(i) + k) \cdot k$. Here, $Z\subseteq S$ contains edges of the form $(v,w)$.  Let us now focus on the conflicts that are resolved by edges of the form $(w,u)$, $u\in U_{i-1}$. This implies that we have paths of the form $\langle v,w,u\rangle$ such that $(v,w)\in E(G)$ and $(w,u)\in S$. The number of paths of this kind is upper bounded by $|S| - |Z| \le k$. This implies that the total number of conflicts $(v,u)$ where $u \in U_{i-1}$, and can be resolved by $S$, is bounded by $(f(i) + k) \cdot k + k = f(i-1)$, by Proposition~\ref{prop:function}. Recall that, by induction, $|U_{i-1}| > f(i-1)$. This contradicts that $S$ is a solution, and the claim follows. 
\end{proof}

Thus, if we find a vertex $w_i$ satisfying the requirements as mentioned in Claim~\ref{cl:nextvertex}, we define $W_i = W_{i-1} \cup \LR{w_i}$ and proceed to the next iteration. Otherwise, if we cannot find $w_i$, then we stop and output the current set $W_{i-1}$ as the set $W_v$ as required by the lemma. In this case, by Claim~\ref{cl:nextvertex}, we know that, if we have a \yes-instance, then $S \cap \LR{ (v, w_j) : j \in [i-1] } \neq \emptyset$. 
  
Suppose the iterative procedure goes on for $\delta$ iterations, i.e., $W_v = W_{\delta} = \LR{w_1, w_2, \ldots, w_{\delta}}$. Observe that if $\delta\geq d$ then the vertex set $\LR{w_1, w_2, \ldots, w_d}$ and  $\bigcap_{j\in[\delta]} ( U\cap N_G(w_j) )$ induce a $\kdd$ in $G$    which contradicts the fact that  $G$ is $\mathcal{K}_{d,d}$-free. Thus, $\delta < d$, and (1) holds. 
\end{proof}

We use the algorithm in Lemma~\ref{lem:kddset} to find the set $W_v$ of size $\delta < d$ as claimed. We know that, if we have a \yes-instance, then any minimal solution $S$ contains at least one edge of the form $(v, w_i)$, $i \in [\delta]$. The following reduction rule essentially tries to guess the set $W' \subseteq W_v$ for which such edges belong to the solution, and add those vertices to the vertex cover $R$. Further, due to our invariant, when we add $W'$ to $R$, we also need to guess the edges between $W'$ and the vertices already in $R$. The following reduction rule is a Turing reduction from the annotated version of the problem to itself -- we prove later that it is a valid decreasing Turing \fpt reduction step.

\begin{red_rule}\label{rr:highdeg}
    Let $v\in R$ be any vertex such that $deg_{\CG}(v)\geq f(0)+1$.  For every non-empty subset $W'\subseteq W_v$, and for any subset 
    $$\overline{E}_{x} \subseteq \Big\{ (p, q) \in \overline{E}(G): p \in W', q \in W' \cup (R \setminus\LR{v}) \Big\}$$ 
    of non-edges  in $G$, we create the following instances. $(G_j,\Gamma,k_j,R')$ where $G_j=G+(\{(v,w_a)|w_a\in W'\} \cup \overline{E}_{x})$, $k_j=k-|W'|-|\overline{E}_{x}|$ and $R'=R\cup W'$. 
    \end{red_rule}


\begin{lemma}\label{lem:redrule4}
\cref{rr:highdeg} is a valid decreasing \fpt-Turing -reduction.
\end{lemma}

\begin{proof}
First, note that $|W_v| \le d$, which implies that the number of subsets $W' \subseteq W$ is at most $2^d$. Then, for each fixed subset $W'$ of size $\delta$, there are at most $2^{\order{k \log d}}$ many subset $\overline{E}_{x}$. So in total  the number of instances we created is bounded by $2^{\order{d+ k \log d}}$.  Now we show the safeness of the \cref{rr:highdeg}.

\begin{sloppypar}
\begin{claim}
    $(G,\Gamma, k,R)$ is a \yes-instance if and only if one of the instances in $\{(G_{j}, \Gamma, k_{j},R')\}$ is a \yes-instance.
\end{claim}
\end{sloppypar}

\begin{proof}
    The backward direction is trivial. In the forward direction, assume that $(G,\Gamma, k,R)$ is a \yes-instance and $S$ is a minimal solution. Due to the  \cref{lem:kddset} the solution $S$ must satisfy the condition that $S \cap \LR{ (v, w_i) : i \in [\delta] } \neq \emptyset$. Let $S' \coloneqq S \cap \LR{ (v, w_i) : i \in [\delta] }$, $W'' \coloneqq V_{S'} \setminus \{v\}$, and $S'' \coloneqq S \cap \{(w,r)~:~w \in W'', r \in R \cup W''\}$. Clearly $W'' \subseteq W_v$ and  $S'' \subseteq \LR{ (p, q) \in \overline{E}(G): \text{either}~p, q \in W' \text{ or } p \in W', q \in R \setminus\LR{v} }$. Now considering $S''= \overline{E}_{x}, G_j= G+ S'+S'', k_j= k-|S'|-|S''|,$ and $R'= R \cup W''$, we have that $(G_{j}, \Gamma, k_{j},R')$ is a \yes-instance.
\end{proof}
As $W'$ is a non-empty subset of $W_v$ so $k_j \leq k-|W'| \leq k-1$. So in each created instance, the value of the parameter decreases by at least one. Hence the lemma follows.    
\end{proof}

Next we formalize an observation that follows from the definition of  \fpt-Turing -reduction.
\begin{observation} \label{obs:r-decrease}
    Starting from an annotated instance $(G,\Gamma, k, R)$, let $(G_j, \Gamma, k_j, R_j)$ be one of the instances produced by \rr~\ref{rr:highdeg}. Then, $|R_j| \le |R| + (k - k_j)$. In particular, this implies that the size of $|R'|$ in any of the instances produced by a sequence of decreasing Turing \fpt reductions, remains bounded by $5k$.
\end{observation}

We keep applying \rr~\ref{rr:highdeg} until each vertex in $R$ has at most $f(0)$ neighbors in $I$ in the graph $\CG$. In each step, we maintain the invariant that we have included all possible solution edges within $R$ in $G$ and reduced $k$ appropriately. Observe that when \rr~\ref{rr:highdeg} is not applicable, the total number of vertices in the conflict is at most $\order{k \cdot f(0)}$.

\begin{mdframed}[backgroundcolor=cyan!5,topline=false,bottomline=false,leftline=false,rightline=false] 
Let  $(G,\Gamma, k,R)$ be an annotated instance of \dtwoa.  Find a dilation $2$-augmentation set $S$ such that $|\{x,y\}\cap R|\leq 1$. That is, every edge added must contain at least one end-point that does not belong to $R$. In addition, $|R|\leq 5k$ and $|V_c|=\order{k \cdot f(0)}$, where $f(0)= d \cdot k^{d} +  k^{d+1} + \lr{2 \cdot \sum_{j = 2}^{d} k^j} + k$.
\end{mdframed}

\subsection{Solving Annotated Instances with bounded $V_c$} \label{sec:solveinstance}

In the rest of the section, we describe how to solve an annotated instance $\mathcal{I} = (G,\Gamma, k, R)$ when the size of $V_c$ is bounded by some function of $k$ and $d$. Note that although $|V_c| \le h(k, d)$, we cannot completely forget about the vertices outside $V_c$ -- since some conflicts may be resolved by using edges incident to vertices outside $V_c$. However, we will argue that we do not need to keep all such vertices, but it suffices to keep one representative from each ``equivalence class'', which we formalize below. Note that in this last step, we do not need the vertex cover $R$.

Let $O = V(G) \setminus V_c$ denote the vertices \emph{outside} $V_c$.  
For each $A, B \subseteq V_c$ with $A \cap B = \emptyset$, let $O(A, B)$ denote the set of vertices $v \in O$ satisfying the following two properties:
\begin{enumerate}
    \item The set of vertices $u \in V_c$ such that $d_G(u, v) = 1$ is \emph{exactly} equal to $A$, and
    \item The set of vertices $w \in V_c$ such that $(v, w) \not\in E(G)$, but $d_{\Gamma}(v, w) = 1$ is \emph{exactly} equal to $B$.
\end{enumerate}
We have the the following observation.

\begin{observation} \label{obs:equiv}
    \begin{enumerate}
        \item $\mathcal{P} = \LR{O(A, B): A, B \subseteq V_c, A \cap B = \emptyset }$ forms a partition of $O$.
        \item $|\mathcal{P}|$ is bounded by $3^{|V_c|} \le g(k, d)$ for some computable function $g$.
    \end{enumerate}
\end{observation}

For each $O(A, B) \in \mathcal{P}$ such that $O(A, B) \neq \emptyset$, we mark an arbitrary vertex $v(A, B) \in O(A, B)$. Let $U \subseteq O$ denote the set of unmarked vertices. In the following reduction rule, we eliminate all the vertices of $U$ from $G$ and $\Gamma$, and then prove that it is correct.

\begin{red_rule} \label{rr:twins}
    Given the annotated instance $\mathcal{I}_1 = (G,\Gamma, k, R)$, produce the annotated instance $\mathcal{I}_2 = (G', \Gamma', k, R)$, where $\Gamma' = \Gamma - U$ and $G' = G - U$.
\end{red_rule}

\begin{lemma} \label{lem:twinrulesound}
    \rr~\ref{rr:twins} is safe.
\end{lemma}
\begin{proof}
    We argue that $\mathcal{I}_1$ is a \yes-instance if and only if $\mathcal{I}_2$ is a \yes-instance. The reverse direction is trivial, since any solution for $\mathcal{I}_s$ is also a solution for $\mathcal{I}_1$ (recall that the vertices of $U$ are not involved in any conflict). So we show the forward direction.

    Let $S \subseteq \overline{E}(G)$ be a minimal solution of size $k$ with the set of end-points being $V_S$. From this solution, we define another solution $S'$ as follows. Consider any $u \in U \cap V_S$, and let $S(u) = \LR{ (u, w) \in S}$, and let $N(u) = \LR{w : (u, w) \in S(u) }$. First, note that $N(u) \subseteq V_c$ -- suppose not, i.e., there exists some $w \in N(u)$ but $w \not\in V_c$, then $(u, w)$ cannot be part of a path of hop length $2$ between two vertices in $V_c$, which contradicts the minimality of $S$. Now, suppose $u$ belongs to the set $O(A, B) \in \cO$, then $O(A, B) \neq \emptyset$, which implies that some $v(A, B) \in V(G) \setminus U = V(\Gamma) \setminus U$. We define another set $S' \coloneqq S \setminus S(u) \cup \LR{ (v(A, B), w : w \in S(u) }$. Note that $|S| = |S'|$. We prove that $S'$ is also a solution.

    Consider any path $\langle x, u, y \rangle$ in $G+S$, such that $d_{\Gamma}(x, u) = d_{\Gamma}(u, y) = 1$. There are three cases: (i) $(x, u) \in E(G)$ and $(u, y) \in S$. In this case, note that $(x, v(A, B)) \in E(G')$ with $d_{G'}(x, v(A, B)) = d_G(x, v(A, B)) = 1$, and $d_{\Gamma}(u, y) = d_{\Gamma}(v(A, B), y) = 1$. Further, $(v(A, B), y) \in S'$. (ii) $(x, u) \in S$ and $(u, y) \in E(G)$ is symmetric. (iii) $(x, y), (y, u) \in S$. In this case, $d_{\Gamma}(v(A, B), x) = d_{\Gamma}(u, x) = 1$, and $d_{\Gamma}(v(A, B), y) = d_{\Gamma}(u, y) = 1$. Thus, in all three cases, $\langle x, v(A, B), y\rangle$ path exists in $G+S'$, or in other words, the replaced edges incident to $v(A, B)$ can resolve the same set of conflicts as the edges of $S(u)$. By iterating over the vertices of $V_S \cap U$ and  modifying the solution in this way, we obtain a solution $S''$ containing entirely the edges within $E(\Gamma')$. This completes the forward direction.
\end{proof}

\begin{corollary} \label{cor:final}
    The number of vertices in the reduced annotated instance $(G', \Gamma', k, R)$ is bounded by some $g(k, d)$.
\end{corollary}

Thus, there are at most $\binom{g(k, d)}{2k}= \binom{2^{\mathcal{O}(h(k, d))}}{2k} = 2^{\order{k \cdot h(k, d)}}$ possibilities for the end-points of the solution edges. For each such subset, we can guess the actual set of at most $k$ edges in time $k^{\order{k}}$. Thus, the total number of possibilities is bounded by $f(k, d)$ for some computable function $f$. We finish the discussion with the following theorem.

\begin{theorem} \label{thm:kddfree}
    {\sc Dilation $2$-Augmentation} can be solved in time $f(k, d) \cdot n^{\order{1}}$ when $G$ is a $\kdd$-free graph for any  $d \in \mathbb{N}$. 
\end{theorem}

The following corollary follows trivially from \Cref{thm:kddfree}.
\begin{corollary} \label{cor:forestplanar}
    {\sc Dilation $2$-Augmentation} can be solved in time $f(k) \cdot n^{\order{1}}$ when $G$ is a forest, planar graph, bounded-treewidth graph, an $H$-minor free graph for some fixed $H$, a nowhere dense graph, or bounded degeneracy graph.
\end{corollary}

     \section{\fpt Algorithms of \dta for Bounded Degree Graphs}\label{sec:BoundDeg}
   In this section either $G$ is of bounded degree or $\Gamma$ is of bounded degree. Observe that a $\kdd$-free graph generalizes bounded-degree graphs, and thus the result when $G$ is of bounded degree may appear to be subsumed by those presented in Section~\ref{sec:kddFree}. However, the result presented here works for any value of $t$, but the result in Section~\ref{sec:kddFree} only worked for $t=2$. 
   
   In either case, we will use the following easy lemma to bound the number of vertices in a ball of radius $\ell$ around a vertex subset of small size. 

    \begin{observation} \label{obs:degreebound}
		Let $H$ be a graph with maximum degree $\Delta$. For any subset of vertices $U \subseteq V(H)$, $|N_H^{\ell}(U)| \le |U| \cdot \sum_{i = 0}^{\ell} \Delta^i \le |U| \cdot \Delta^{\ell+1}$. 
		\label{obs:elldist}
    \end{observation}

    
     \subsection{$\Gamma$ is of Bounded Degree}\label{sec:boundgen}
      Let $(G,\Gamma,k)$ be an instance of \dta.  In this subsection, we consider the case where $\Gamma$ belongs to the family of graphs of maximum degree at most $\Delta$. However, there is no restriction on $G$. That is, $G$ is an arbitrary undirected graph. 
The idea of the proof is to identify a subset of vertices of $\Gamma$ of size at most $f(k,\Delta,t)$ such that $V_S$ belongs to balls of radius $t$ around them. Once the set is identified, the algorithm tries all potential solutions of size at most $k$ and returns \yes, if either leads to the desired solution. 
    
   The hypothetical solution $S$ is of size at most $k$, and hence from Observation~\ref{obs:elldist} we have $|N^t_{\Gamma}(V_S)|\leq 2k\cdot \Delta^{t+1}$. First, we show that the vertices of $V_S$ are in distance $t$ from the vertices of $V_c$ in $\Gamma$.

\begin{lemma}\label{lemma:boundGamma}
Let $(G,\Gamma,k)$ be a \yes-instance of \dta. Then, $V_S\subseteq N^t_{\Gamma}(V_c)$.
\end{lemma}
	\begin{proof}
		Let $S$ be a minimal solution of \dta on an instance $(G,\Gamma,k)$. Furthermore, let $u$ be a vertex in $V_S$. As $u \in V_S$, there exists a vertex $v$ in $V_S$ such that $(u,v) \in S$. Consider $S'=S\setminus \{(u,v)\}$. Observe that $S'$ is not a solution, as otherwise it will contradict the minimality of $S$. 
        Therefore, there exist two vertices $x,y\in V_c$ such that $x$ and $y$ are in conflict (in fact, adjacent conflict) in $G+S'$, but not in conflict in $G+S$. Thus, there exists a (weighted) shortest path $P$ in $G+S$ between $x$ and $y$ that contains $(u,v)$ and $d_{G+S}(x,y)\leq t$. Thus, the subpath $P_{xu}$ of $P$ has length at most $t$ in $G+S$ (that is, $d_{G+S}(x,u)\leq t$). This together with \cref{obs:path} implies that there exists a path $P'$ of length at most $t$ in $\Gamma$ between $x$ and $u$. This implies $u \in N^t_{\Gamma}(x) \subseteq  N^t_{\Gamma}(V_c)$. This concludes the proof.     
	\end{proof}

Observe that we cannot upper bound the size of $N^t_{\Gamma}(V_c)$, as the size of $V_c$ may not be bounded by any function of $t,k$ and $\Delta$. Next, we show a kind of converse of Lemma~\ref{lemma:boundGamma} showing that, in fact, $V_c$ is contained inside the balls of radius $t$ around $V_S$ in $\Gamma$. The proof is identical to Lemma~\ref{lemma:boundGamma} and is presented separately for clarity. This will allow us to obtain our algorithm.

\begin{lemma}\label{lem:coverc}
Let $(G,\Gamma,k)$ be a \yes-instance of \dta. Then,  $V_c\subseteq N^t_{\Gamma}(V_S)$.
\end{lemma}
\begin{proof}
Let $S$ be a minimal solution of \dta on an instance $(G,\Gamma,k)$. Furthermore, let $x\in V_c$ be an arbitrary vertex that is in conflict (again, adjacent conflict)  with some vertex $y$ in $G$. Since $S$ is a solution, we have $d_{G+S}(x,y)\leq t$. Thus, there exists a (weighted) shortest path $P$ in $G+S$ between $x$ and $y$ that contains some edge of the solution, say $(u,v) \in S$. Thus, the subpath $P_{xu}$ of $P$ has length at most $t$ in $G+S$ (that is, $d_{G+S}(u,x)\leq t$). This together with \cref{obs:path} implies that there exists a path $P'$ of length at most $t$ in $\Gamma$ between $u$ and $x$. This implies $x \in N^t_{\Gamma}(u) \subseteq  N^t_{\Gamma}(V_S)$. This concludes the proof.      
\end{proof}

Lemma~\ref{lem:coverc}, along with Observation~\ref{obs:degreebound} implies that in a \yes-instance, the size of $V_c$ is bounded by $2k \cdot \Delta^{t+1}$. Thus, if $|V_c| > 2k \cdot \Delta^{t+1}$, we say \no. Now assume that $|V_c| \le 2k \cdot \Delta^{t+1}$. Next, via Lemma~\ref{lemma:boundGamma}, we have that $V_S \subseteq N^t_\Gamma(V_c) =: Q$, say. Again, by Observation~\ref{obs:degreebound}, $|Q| \le |V_c| \cdot \Delta^{t+1} \le 2k \cdot \Delta^{2t+2}$.

 
	We know we have to add at most $k$ edges between the vertices in $Q$ to resolve all the conflicts. This can be done by first guessing the set of end-points, which is a set of vertices of size at most $2k$, and for each such set of end-points, trying all $k^{\order{k}}$ possibilities for different edges that can be added. Overall, the number of guesses can be upper bounded by $\binom{|Q|}{2k} \cdot k^{\order{k}} \le 2^{\order{k \log k}} \cdot \Delta^{\order{kt}}$. Hence we have the following theorem from Observation~\ref{obs:elldist} and Lemma~\ref{lem:coverc}.

	\begin{theorem}\label{theorem:boundedGamma}
		{\sc Dilation $ t$-Augmentation}   can be solved in time $2^{\order{k \log k}} \cdot \Delta^{\order{kt}} \cdot n^{\order{1}}$, where $\Delta$ denotes the maximum degree of the graph $\Gamma$.
	\end{theorem}

 \subsection{$G$ is of Bounded Degree}\label{sec:genbound}
Let $(G,\Gamma,k)$ be an instance of \dta.  In this subsection, we consider the case where $G$ belongs to the family of graphs of maximum degree at most $\Delta$. However, there is no restriction on $\Gamma$. That is, $\Gamma$ is an arbitrary undirected graph. The proof strategy in this section is similar to that employed for Theorem~\ref{theorem:boundedGamma}. As before, the idea is to identify a subset of vertices of $G$ of maximum size $f(k,\Delta,t)$ such that $V_S$ belongs to balls of radius $t^2$ around them.

   \begin{lemma}\label{lemma:vcnotlarge} 
	Let $(G,\Gamma,k)$ be a \yes-instance of \dta. Then,  $V_c\subseteq N^t_{G}(V_S)$. 
	\end{lemma}
	\begin{proof}
 Let $S$ be a minimal solution of \dta on an instance $(G,\Gamma,k)$. Furthermore, let $x\in V_c$ be an arbitrary vertex that is in conflict (again, adjacent conflict)  with some vertex $y$ in $G$. Since $S$ is a solution, we have $d_{G+S}(x,y)\leq t$. Thus, there exists a (weighted) shortest path $P$ in $G+S$ between $x$ and $y$ that contains an edge of $S$ and hence a vertex of $V_S$. If $x \in V_S$, then we are done. Otherwise, let $u$ be the vertex in $P$ such that $u \in V_S$ and that no internal vertex from the subpath $P_{xu}$ of $P$ contains a vertex from $V_S$. Notice that the subpath $P_{xu}$ of $P$ contains no edge from $S$, so it has maximum length $t$ in $G$ (that is, $d_{G}(u,x)\leq t$).  This implies $x \in N^t_{G}(u) \subseteq  N^t_{G}(V_S)$. This concludes the proof. 
	\end{proof}
	
	As $ |V_S| \leq 2k $, the next observation follows from Observation~\ref{obs:elldist}. 
	\begin{observation}
		$|N^t_G(V_S)|\leq 2k\cdot \Delta^{t+1}$.\label{obs:bound1}
	\end{observation}

 Lemma~\ref{lemma:vcnotlarge} and Observation~\ref{obs:bound1} together yield the following reduction rule that yields an upper bound on the size of the conflict set $V_c$.

\begin{red_rule}\label{rule:sizeVc}
If $|V_c| > 2k\cdot \Delta^{t+1}$, return \no.
\end{red_rule}
	
To identify the vertices in $V_S$, we next show that they lie in the balls of radius $t^2$ around $V_c$ in $G$. 
	
\begin{lemma}\label{obs:notfar}
Let $(G,\Gamma,k)$ be a \yes-instance of \dta and $S$ be a hypothetical minimal solution. 
Then, $V_S\subseteq N^{t^2}_G(V_c)$.
\end{lemma}

 \begin{figure}[t]
            \centering
            \includegraphics[width=1\textwidth]{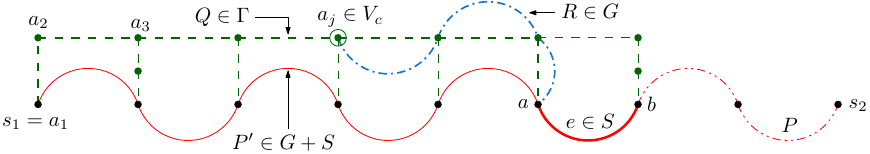}
            \caption{Path in $G+S$ is shown in red (solid lines), path in $\Gamma$ is shown in green (dashed) and path in $G$ is shown in blue ( dashed-dotted).}
            \label{fig:deg}
        \end{figure}
	
	\begin{proof}
 Let $a \in V_S$ and $(a,b)$ be an edge in $S$ containing $a$. 
	Due to the minimality of $S$, for every edge $e=(a,b)$ in $S$, there exist at least two vertices $s_1$ and $s_2$ such that (i) $s_1$ and $s_2$ are in conflict (again, adjacent conflict) with each other in $G$, and (ii) $e=(a,b)$ appears in a shortest path $P$ of length at most $t$ between $s_1$ and $s_2$ in $G+S$. Note that $a$ and $b$ may not be in conflict. We show that there exists a vertex $w\in V_c$ such that $w$ is at most $t^2$ hop distance from the vertex $a$ in $ G $ (that is, ${\sf hop}_G(w,a) \leq t^2$). 
 Let $P'$ be the subpath of length $\tau\leq t$ between $s_1$ and $a$ (see Figure~\ref{fig:deg}). Since $G+S$ is also a graph embedded in a metric space derived from the undirected graph $\Gamma$, by Observation~\ref{obs:path} we have $d_{\Gamma}(s_1,a) \leq d_{G+S}(s_1,a)=\tau \leq t$. 
 
 Let $Q= \langle s_1=a_1,a_2\cdots , a_{\tau+1}=a\rangle$ be a path between $s_1$ and $a$ in $\Gamma$ of length $d_{\Gamma}(s_1,a) \leq t$. Let $j\in[\tau+1]$ be the largest index such that $a_j\in V_c$. Since $s_1\in V_c$, such an index $j$ always exists. Observe that for all $j\leq i\leq \tau$, $a_i$ and $a_{i+1}$ are {\em not in conflict} and $a_i$ and $a_{i+1}$ are {\em adjacent} in $\Gamma$. Therefore, there exists a path $R_i$ between $a_i$ and $a_{i+1}$, $j\leq i\leq \tau$, in $G$ of maximum length $t$. Thus, by concatenating the paths $R_j  \bullet R_{j+1} \bullet \cdots \bullet R_\tau$, we get a walk with a weighted length at most $t^2$ between $a_j$ and $a$ in $G$. This implies that there is a path between $a_j$ and $a$ in $G$ of length $t^2$. That is, $d_{G}(a_j,a)\leq t^2$. Since all the edge weights in $G$ are at least $1$, we have ${\sf hop}_G(a_j,a)\leq t^2$. This implies $a \in N^{t^2}_{G}(a_j) \subseteq  N^{t^2}_{G}(V_c)$. This concludes the proof. 
    \end{proof}
   
  Since $V_S\subseteq N^{t^2}_G(V_c)$ (Lemma~\ref{obs:notfar}) and $|V_c| \leq  2k\cdot \Delta^{t+1}$ (Reduction Rule~\ref{rule:sizeVc}), by observation~\ref{obs:elldist} we get the following. 
	\begin{lemma}
 Let $(G,\Gamma,k)$ be a \yes-instance of \dta. Then, 
		$|N^{t^2}_G(V_c)|\leq 2k\cdot \Delta^{t+1}\cdot \Delta^{t^2+1}$.
	\end{lemma}
	
	We know we have to add at most $k$ edges between the vertices in $N^{t^2}_G(V_c)$ to resolve all the conflicts. This can be done by first guessing the set of end-points, which is a set of vertices of size at most $2k$, and for each such set of end-points, trying all $k^{\order{k}}$ possibilities for different edges that can be added. Overall, the number of guesses can be upper bounded by $\binom{|N^{t^2}(V_c)|}{2k} \cdot k^{\order{k}} \le 2^{\order{k \log k}} \cdot \Delta^{\order{kt^2}}$. Hence we have the following theorem.

	\begin{theorem}\label{theo:boundedG}
		{\sc Dilation $ t$-Augmentation}   can be solved in time $2^{\order{k \log k}} \cdot \Delta^{\order{kt^2}} \cdot n^{\order{1}}$, where $\Delta$ denotes the maximum degree of the graph $G$.
	\end{theorem}

      \section{Dilation Augmentation For Forest and Stars}\label{sec:hardTree}


In this section, we focus on \dta when $G$ is a disjoint union of a star, and a set of isolated vertices (Section~\ref{sec:genstar3}), or when $\Gamma$ is a star (Section~\ref{sec:stargen3}). We show that \dtai{3} is {\sf W}-hard in both cases. In contrast to this, we know that \dtai{2} is \fpt when $G$ is a forest, due to Theorem~\ref{thm:kddfree}. Finally, in Section~\ref{sec:treegen2},  we will observe that \dtai{2} is polynomial-time solvable when $\Gamma$ is a tree, and $G$ is an arbitrary graph.
 

\subsection{{\sf W[1]}-hardness of \dtai{3} when $G$ is   Forest}\label{sec:genstar3}
	In this section, we show that {\sc Dilation $3$-Augmentation} is \woh~even when $G$ is a forest, in particular, a disjoint union of a star and an independent set. We give a polynomial-time parameter preserving reduction from the {\sc Multicolored Clique} problem which is known to be \woh~\cite{fellows2007fixed}. The input of {\sc Multicolored Clique} consists of a graph $H$, an integer $ k $, and a partition $\V= (V_1, V_2, \cdots, V_k) $ of the vertices of $ H$; our aim is to decide if there is a clique of size $k$ containing exactly one vertex from each set $ V_i, i \in [k] $. We denote this instance by $(H,k,\V)$. We can assume that for each $ i \in k $, $V_i $ is an independent set.

	\begin{figure}[t!]
		\centering
		\includegraphics[width=1\textwidth]{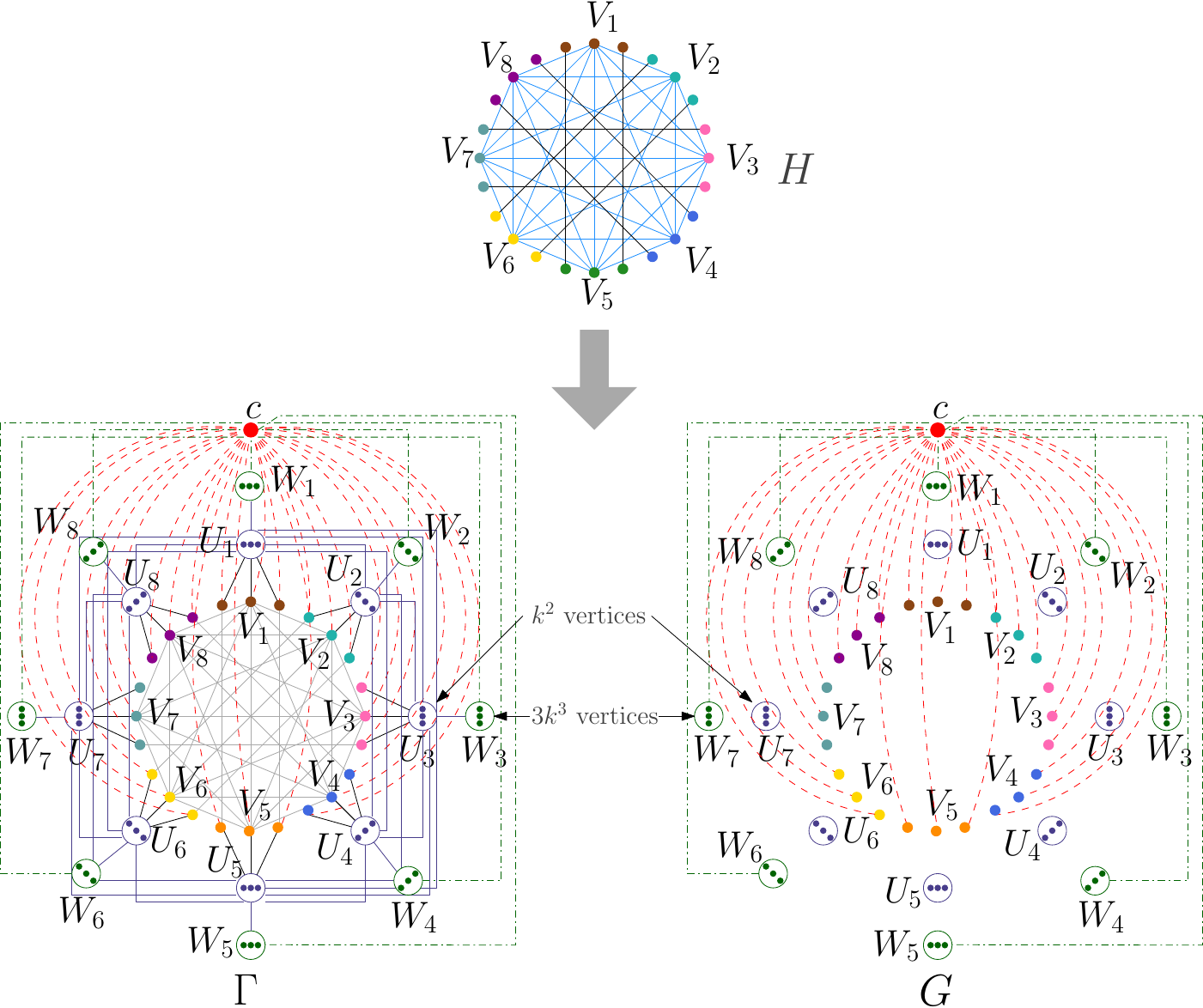}
				\caption{Hardness for $t=3$ when $G$ is a disjoint union of a star, and a set of isolated vertices.\label{fig:hard1}}
	\end{figure}
	
	From an instance $(H,k,\V)$ of {\sc Multicolored Clique} for $k\geq 2$, we construct an instance $ (G, \Gamma, k') $ of {\sc Dilation $3$-Augmentation} with the following way (see Figure \ref{fig:hard1}).
	
	\begin{itemize}
	   
		\item For every set $ V_i, i \in [k]$, we introduce a set $ U_i $ of $k^2$ many vertices and a set $W_i$ of $3k^3$ many vertices. 
		
		\item Construction of the graph $ \Gamma $: 
		
		\begin{enumerate}
		    		
			\item The vertex set in $\Gamma$ consists of vertex set $H_G\cup \{U_1\cup\cdots 	\cup U_k\} \cup \{W_1\cup\cdots 
			\cup W_k\}\cup\{c\}$ where $H_G=V(H)$.
			\medskip
			
			\item The edge set of $\Gamma$ is defined by $E(\Gamma)=E_H\cup E_U\cup E_W\cup E_c$; here $E_H=E(H)$, $E_U=\{(v,u)|v\in V_i,u\in U_i, i\in[k]\} \bigcup \{(v,u)|v\in U_i,u\in U_j, i,j\in[k], i \neq j\} $, $E_W=\{(u,w)|u\in U_i,w\in W_i, i\in[k]\}$ and $E_c=\{(c,v)|v\in H_G\cup W_1 \cup \ldots \cup W_k\}$.
		\end{enumerate}
		
		\medskip 
		\item Construction of the graph $ G$: The graph $G$ consists of the vertex set $V(\Gamma)$ and the edge set $E_c$. 
		
		\medskip
		\item We set $ k'= {k\choose 2}+k^3 $.
	\end{itemize}


	It is easy to see that in the constructed instance $ (G, \Gamma, k') $ the graph  $ G $ is indeed a disjoint union of a star and independent set. Also, the construction can be performed in time polynomial in $|V(H)|$ and $ k $. Towards the correctness of our reduction, we prove the following.

	\begin{lemma}\label{lem:reduction1}
		$ (H,k,\V) $ is a \yes-instance of {\sc Multicolored Clique} if and only if $ (G,\Gamma, k') $ is a \yes-instance of {\sc Dilation $3$-Augmentation}.
		\end{lemma}
	\begin{proof}
		$ (\Rightarrow) $In the forward direction, suppose that there is a multicolored clique $Q$ in $H$. 
		Let the set of vertices in $Q$ be $V_Q=\{v_1,\cdots,v_k\}$ where $v_i\in V_i$ and the edgeset of $Q$ be $E_Q$. Consider the set of edges $E^{VU}=\{(v_i,u)|i\in[k], u\in U_i\}$. Observe that $|E_Q|={k\choose 2}$ and $|E^{VU}|=k \cdot k^2$ as $ |U_i|=k^2 $ for each $ i \in [k] $. Next, we show that the dilation of $G+(E_Q\cup E^{VU})$ is at most $3$.	Notice that all the edges incident to $c$ in $ \Gamma $ are also present in $G$. Hence $c$ is not involved in any adjacent conflicts. Now we consider all other pairs of adjacent vertices in $ \Gamma $, say $ a $ and $ b $, and show that there exists a  path between them of length at most 3 in $G+(E_Q\cup E^{VU})$. This suffices our proof due to \cref{lem:adjacent}. We argue with the following cases.

		\begin{description}
			\item[{Case (i)}: $a\in V_i$ and $b\in V_j$.] There is a length $2$ path between them in $ G $ itself which is precisely $\langle a,c,b \rangle$.
			\medskip 
			\item[{Case (ii)}: $a\in V_i$ and $b\in U_i$.] If $a=v_i$ then $a$ and $b$ are adjacent in $ G+E^{VU} $. Otherwise, there is a length $3$ path between them in $ G+E^{VU} $ which is precisely $\langle a,c,v_i,b \rangle$.
			\medskip  
			\item[{Case (iii)}: $a\in U_i$ and $b\in U_j$.] If $i\neq j$, we use two edges from $E_Q\cup E^{VU}$. The vertices $a$ and $b$ are connected by a length $3$ path $\langle a,v_i,v_j,b \rangle$ in $G+(E_Q\cup E^{VU})$. Suppose that $i=j$. Then $\langle a,v_i,b \rangle$ is a path in $G+(E_Q\cup E^{VU})$ of length 2. 
            \medskip  
			\item[{Case (iv)}: $a\in U_i$ and $b\in W_i$.] In this case, we use an edge from $E^{VU}$. The vertices $a$ and $b$ are connected by a length $3$ path $\langle a,v_i,c,b \rangle$ in $G+E^{VU}$.
		\end{description}

		\noindent $ (\Leftarrow) $ In the backward direction, let $ E_X $ be a solution to the instance $  (G,\Gamma, k') $ of {\sc Dilation $3$-Augmentation}, that means dilation of $G+E_X$ is at most $3$.  
		%
        %
        We start with the following claim.
        
		\begin{claim}\label{claim:uvertex}
			For each $ i \in [k] $ and every $u\in U_i$, the set $ E_X $ contains an edge $(u,v)$ for $v\notin \bigcup_{j=1}^kU_j$.
		\end{claim}
		
		\begin{proof}
			The proof is by contradiction. Suppose that there are $i \in [k] $ and $u\in U_i$ such that for every $(u,v)\in E_X$, $v\in \bigcup_{j=1}^kU_j$.  
            Because $|E_X|\leq k'=\binom{k}{2}+k^3$ and $|W_i|=3k^3$, we have that $|W_i|-2|E_X|\geq 1$ and, therefore, there is $w\in W_i$ such that  $w$ is not incident to any edge of $E_X$. Thus, any shortest $(u,w)$-path in $G+E_X$ contains the edge $(w,c)$ and an edge $(u,v)$ for some $v\in \bigcup_{j=1}^kU_j$. However, the distance between $c$ and $v$ in $\Gamma$ is 2. This implies that the length of any $(u,w)$-path in $G+E_X$ is at least 4. Because $u$ and $w$ are adjacent in $\Gamma$, this contradicts that the dilation of $G+E_X$ is at most 3 and proves the claim.          
		\end{proof}

		\noindent Due to the above claim, summing over all $i \in [k]$, we have that the set $ E_X $ contains at least $ k \cdot k^2=k^3$ edges incident to the vertices of $\bigcup_{i=1}^k U_i $ whose second end-points are outside $\bigcup_{i=1}^k U_i $. Because $|E_X|\leq \binom{k}{2}+k^3$, we have that for each $i\in[k]$, there are at least $k^2-2\binom{k}{2}\geq k\geq 2$ vertices of $U_i$ that are adjacent to exactly one edge of $E_X$. 
        For $i\in[k]$, we denote by $U_i'\subseteq U_i$ the set of vertices incident to unique edges of $E_X$. 
        We make the following observation.

        \begin{claim}\label{claim:adj}
			For each $ i\in [k] $, there is $u\in U_i'$ such that $u$ is incident to a single edge $(u,v)\in E_X$ such that $(u,v)\in E(\Gamma)$.
		\end{claim}

        \begin{proof}
            Consider arbitrary $u\in U_i'$. This vertex has a unique neighbor $v$ in $G+E_X$ where $v\notin\bigcup_{j=1}^kU_i$ by Claim~\ref{claim:uvertex}. If $(u,v)\in E(\Gamma)$, the claim holds. Suppose that $(u,v)\notin E(\Gamma)$. Recall that $|U_i'|\geq 2$. Thus, there is $u'\in U_i'$ distinct from $u$. In the same way as above, $u'$ has a unique neighbor $v'$ in $G+E_X$ where $v'\notin\bigcup_{j=1}^kU_i$. Then any $(u,u')$-path in $H+E_X$ contains the edges $(u,v)$ and $(u',v')$. Since $(u,u')\in E(\Gamma)$ and the dilation of $G+E_X$ is at most 3, we have that the length of $(u',v)$ is one. Therefore, $(u',v')\in E(\Gamma)$. This concludes the proof.
        \end{proof}

        Using Claim~\ref{claim:adj}, for every $i\in[k]$, we denote by $u_i$ the vertex of $U_i'$ that is  incident to a single edge of $E_X\cap E(\Gamma)$ and we use $v_i$ to denote the second end-point of the edge. 

        \begin{claim}\label{claim:clique}
			$\{v_1,\ldots,v_k\}$ is a clique of $H$.
		\end{claim}

        \begin{proof}
        First, we show that $v_i\in V_i$ for each $i\in[k]$. Consider $i\in[k]$. Note that by Claims~\ref{claim:uvertex} and \ref{claim:adj} and the construction of $\Gamma$, either $v_i\in W_i$ or $v_i\in V_i$. Suppose that $v_i\in W_i$ and consider $u_j$ for $j\in[k]$ distinct from $i$. We have that $(u_i,u_j)\in E(\Gamma)$. Therefore, $G+E_X$ should have a $(u_i,u_j)$-path of length at most 3. Any 
        $(u_i,u_j)$-path contains the edges $(u_i,v_i)$ and $(u_j,v_j)$. Thus, a shortest $(u_i,u_j)$-path should contain $(v_i,v_j)\in E(\Gamma)$. However, since $v_i\in W_i$ and $v_j\in V_j\cup W_j$, we have no such an edge by the construction of $\Gamma$. This contradiction proves that $v_i\in V_i$.

        Finally, we show that for all distinct $i,j\in[k]$, $v_i$ and $v_j$ are adjacent in $H$. For this, we again observe that $G+E_X$ should have a $(u_i,u_j)$-path of length at most 3 that includes $(u_i,v_i)$ and $(u_j,v_j)$. This implies that $(v_i,v_j)\in E(\Gamma)$. Because $v_i\in V_i$ and $v_j\in V_j$, we have that $(v_i,v_j)$ is in $E_X$ and is an edge of $H$. This proves the claim. 
        \end{proof}

        Claim~\ref{claim:clique} completes the proof of the lemma.
        \end{proof}

	Hence, we have the following theorem.

	\begin{theorem}\label{theo-hardtree}
{\sc Dilation $3$-Augmentation} is \woh~ parameterized by $ k $  when $ G $ is  star forest. 
	\end{theorem}


\subsection{{\sf W[1]}-hardness of \dtai{3} when $\Gamma$ is  Star}\label{sec:stargen3}

We prove that \dtai{3} is \wth when $\Gamma$ is a Star graph by presenting a parameter preserving reduction from the {\sc Dominating Set}  problem.

Let $(H,k)$ be an instance of {\sc Dominating Set}. We create the instance $(G, \Gamma, k)$ of the \dtai{3} problem as follows. Graphs $\Gamma$ and $G$ are defined over the vertex sets $V(H) \cup \LR{c}$, where $c$ is a new vertex distinct from the vertices of $V(H)$. $\Gamma$ is a star with center $c$ and leaves $V(H)$. That is, $E(\Gamma) = \LR{ (c, v): v \in V(H) }$. Note that in the shortest path metric defined by $\Gamma$, $d_{\Gamma}(u, v) = 2$ for any $u, v \in V(H)$.

Next, the set of edges of $G$ is the same as that of $H$, i.e., $E(G)=E(H)$, which implies that $G$ is a disjoint union of $H$ and the singleton vertex $c$ (see Figure \ref{fig:stargen} for an illustration of the construction). Note that the weight of each edge in $G$ is $2$.

\begin{figure}[t]
	\centering
	\includegraphics[width=.8\textwidth]{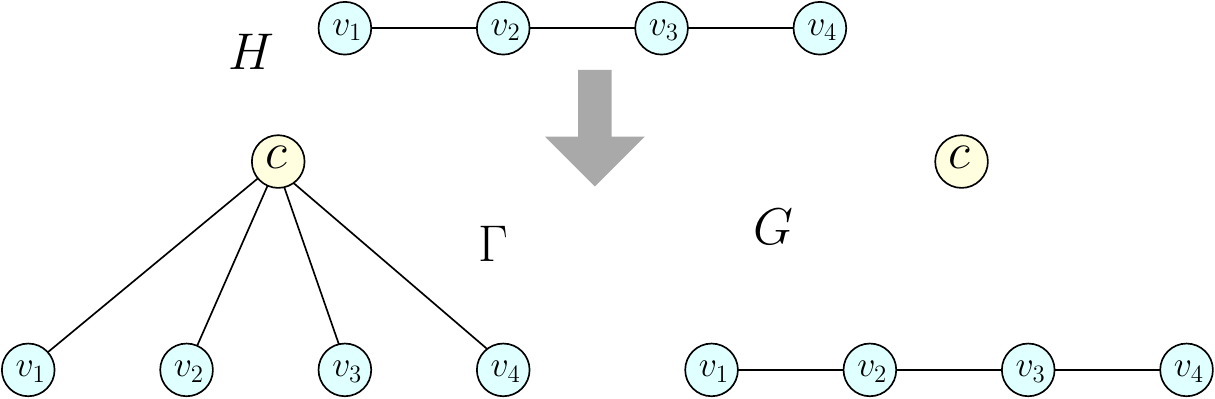}
	\caption{Hardness for $t=3$ when $\Gamma$ is a star.}
 \label{fig:stargen}
\end{figure}

\begin{observation}
    $(H,k)$ is an \yes-instance of {\sc Dominating Set} if and only if $(G, \Gamma, k)$ is an \yes-instance of \dtai{3}.
\end{observation}
\begin{proof}
    Let $D\subseteq V(H)$ be a dominating set in $H$. Let $E_D \coloneqq \{(c,v) : v \in D\}$. We claim that the dilation of $G+E_D$ is at most $3$. 

    For any vertex $v\in V(H)$, $d_{\Gamma}(c,v)=1$. If $v \in D$, then $d_{G+E_D}(c, v) = 1$. Otherwise, if $v \not\in D$, then there exits some $u \in D$ such that $(u, v) \in E(H) = E(G)$. Further, since $(u, c) \in E_D$, it follows that $d_{G+E_D}(v, c) = 2$.

    Now consider any pair of vertices $u, v \in V(H)$. Note that $d_{\Gamma}(u, v) = 2$. We want to show that $d_{G+E_D}(u, v) \le 6$. Since $D$ is a dominating set, there exists some $u' \in N_H[u] \cap D = N_G[u] \cap D$ (resp. $v' \in N_H[v] \cap D = N_G[v] \cap D$). Note that $u'$ may be same as $u$ (resp. $v'$ may be same as $v$). In any case, $(u', c), (v', c) \in E_D$. This implies that, $d_{G+E_D}(u, v) \le d_{G}(u, u') + d_{\Gamma}(u', c) + d_{\Gamma}(c, v') + d_{G}(v', v) \le 2 + 1 + 1 + 2 = 6$. 

    In the reverse direction, suppose $(G, \Gamma, k)$ is a \yes-instance of \dtai{3}, and let $E_C$ be a minimal solution. Suppose there exists some edge $(u, v) \in E_C$ such that $u, v \in V(H)$, and $(u, c) \in E_C$. Then note that $d_{\Gamma}(u, v) = 2$. Now consider an alternate solution $E'_C = E_C \setminus \LR{ (u, v) } \cup \LR{ (v, c) }$. In this new solution, $d_{G+E'_C}(u, v) = 2$, and any path in $G+E_C$ that used the edge $(u, v)$ can be replaced by the subpath $\langle u, c, v \rangle$ of the same length. By performing this kind of replacements for each edge $(u, v) \in E_C$, we obtain a solution, say $E'_C$ with the property that if $(u, v) \in E'_C$, then $(u, c) \not\in E'_C$ and $(v, c) \not\in E'_C$.

    Let $V_C$ be the set of neighbours of $c$ in $G+E'_C$. Observe that as $c$ is an isolated vertex in $G$, and we have added at most $k$ edges of $E'_C$, $|V_C|\leq k$. We claim that $V_C$ is a dominating set of $H$. Suppose there exists some $u \in V(H)$ such that $N_H[u] \cap V_c = \emptyset$. Then, for each $v \in V_c$, it holds that $(u, v) \not\in E'_C$ -- otherwise, we would have edges $(u, v) \in E'_C$ and $(v, c) \in E'_C$, a case handled by the replacement procedure. Therefore, any path in $G+E'_C$ between $u$ and $c$ must use at least two edges of $E(G)$, plus an edge of the form $(w, c)$ for some $w \in V_C$. Hence, such a path has length at least $5$. However, this implies that $d_{G+E'_C}(u, c) \ge 5 > 3 = 3 \cdot d_{\Gamma}(u, c)$, which is a contradiction. Hence, $V_C$ is a dominating set for $H$ of size at most $k$.
\end{proof}
Hence we have the following theorem.

\begin{theorem}
    \dtai{3} is \wth~parameterized by $k$ when $\Gamma$ is  star.
\end{theorem}

\subsection{\dtwoa when $\Gamma$ is Tree}\label{sec:treegen2}

We prove that the \textsc{Dilation $2$-Augmentation} problem is polynomial-time solvable when $\Gamma$ is a tree. Consider any instance of \textsc{Dilation $2$-Augmentation} problem $\I=(G, \Gamma, k)$ where $\Gamma$ is a tree. For any solution $S$ for $\I$ we have the following claim.

\begin{observation} \label{obs:gammaTree}
    For any solution $S$ for $\I$, $E(\Gamma)\subseteq E(G+S)$.
\end{observation}
\begin{proof}
    Suppose for the sake of contradiction that there exists an edge $(u,v)\in E(\Gamma)$, but $(u, v) \not\in E(G+S)$. Observe that $d_{\Gamma}(u,v)=1$, and since $S$ is a solution, $d_{G+S}(u,v)\leq 2$. First, $d_{G+S}(u, v) \neq 1$, since the only way to achieve this is by the direct edge $(u, v)$, which, by assumption, does not exist in $G+S$. Thus, $d_{G+S}(u, v) = 2$. Consider the shortest path $\langle u, w, v \rangle$ of hop and length $2$, going via some other vertex $w$. Therefore, $d_{G+S}(u, w) = d_{\Gamma}(u, w) = 1$, and $d_{G+S}(w, v) = d_{\Gamma}(w, v) = 1$. This implies that the edges $(u, w)$ and $(v, w)$ are also present in $\Gamma$. However, these two edges, along with the edge $(u, v)$ induce a cycle in $\Gamma$, which is a contradiction.   
\end{proof}

Given an instance $(G, \Gamma, k)$ of \dtwoa where $\Gamma$ is tree, let Let $A = E(G) \setminus E(\Gamma)$. Observation~\ref{obs:gammaTree} implies that any solution must contain all the edges of $E(\Gamma)$.  hus, if $k < |A|$, then we say \no. Otherwise, we output $G+A$. Note that $G+A$ contains all the edges of $\Gamma$, and hence has dilation $1$. 
Thus we have the following theorem.
\begin{theorem}
    \dtwoa is polynomial-time solvable when $\Gamma$ is   tree.
\end{theorem}
     \section{{\sf W[2]}-hardness of \dtai{(2+\epsilon)} when $\Gamma$ is Weighted} \label{Sec:hard_weight}

An edge-weighted graph is said to be $(1,w)$-weighted graph if the weight of each edge is either $1$ or $w$. We show in the earlier section that when either $ G $ or $ \Gamma $ is of bounded degree $ \Delta $ then \textsc{Dilation $t$-Augmentation}  admits   \fpt~algorithm parameterized by $ k+t+\Delta $. This section considers the case when $ \Gamma  $ is a weighted graph and $ G $ is of bounded degree $ \Delta $. Here we show this problem is unlikely to admit \fpt~algorithm parameterized by $ k+t+\Delta $. More strongly it admits no XP algorithm parameterized by $ t+\Delta $~even when $\Gamma$ is an $(1,w)$-weighted graph. More formally, we prove the following theorem.

\begin{theorem}\label{theo:weighted}
For any $ \epsilon \in (0,1) $, the {\sc Dilation $(2+\epsilon)$-Augmentation} parameterized by  $ k $ is \wth~when $\Gamma$ is  $(1,w)$-weighted and $ G $ is  subcubic. 
\end{theorem}

\begin{proof}
Towards proving that the 	weighted dilation is \wth, we give a polynomial-time parameter preserving reduction 
from {\sc Diameter-$2$ Augmentation}  parameterized by solution size to \textsc{Dilation $(2+\epsilon)$-Augmentation} problem.
 In {\sc Diameter-$2$ Augmentation}, we are given a graph $H=(V,E)$, and a positive integer $k$, the aim is to decide whether there exists a set $E_1\subseteq E$ of size at most $k$ such that the graph $H_1 = ( V, E \cap E_1 )$ has diameter 2.
  It is well known that {\sc Diameter-$2$ Augmentation} is \wth~parameterized by the solution size $k$ \cite{DBLP:journals/dam/GaoHN13}. Consider an instance $(H,k)$ of {\sc Diameter-$2$ Augmentation}. First, we describe the construction of an instance $(G,\Gamma,k)$ of \dla. Let $V(H)=\{v_1,\dots,v_n\}$. To construct the instance  $(G,\Gamma,k)$ we apply the following procedure (see Figure \ref{fig:wighted} for an illustration of the construction). 

\begin{enumerate}
	\item We set $ w= \frac{3n}{2 \epsilon} $.
	\medskip
	\item 	We create $\Gamma$ with vertex set $\{v_i^j~|~ i,j \in [n]\}$. The edge set of $\Gamma$ consists of  $ E_1, E_2, E_3;$ here  $ E_1 = \{(v_i^j, v_{i'}^{j'})~|~  i,j,i',j' \in [n], j \neq j'\}, E_2= \{(v_i^j, v_{i+1}^{j})~|~ i \in [n-1],  j \in [n]\}, $ and $ E_3= \{(v_1^j, v_n^{j})~|~j \in [n]\}$. In each edge of $ E_1 $, we assign the weight $ w $. We put the weight one to all other edges in $ \Gamma $.
	\medskip 
	
	\item We construct $G$ with vertex set $V(\Gamma)$ and edge set $E_2, E_3, E_4;$ here $ E_4= \{ (v_i^j, v_j^i)~|~ (v_i, v_j) \in E(H) \} $.
\end{enumerate}

\begin{figure}[t]
	\centering
	\includegraphics[width=.8\textwidth]{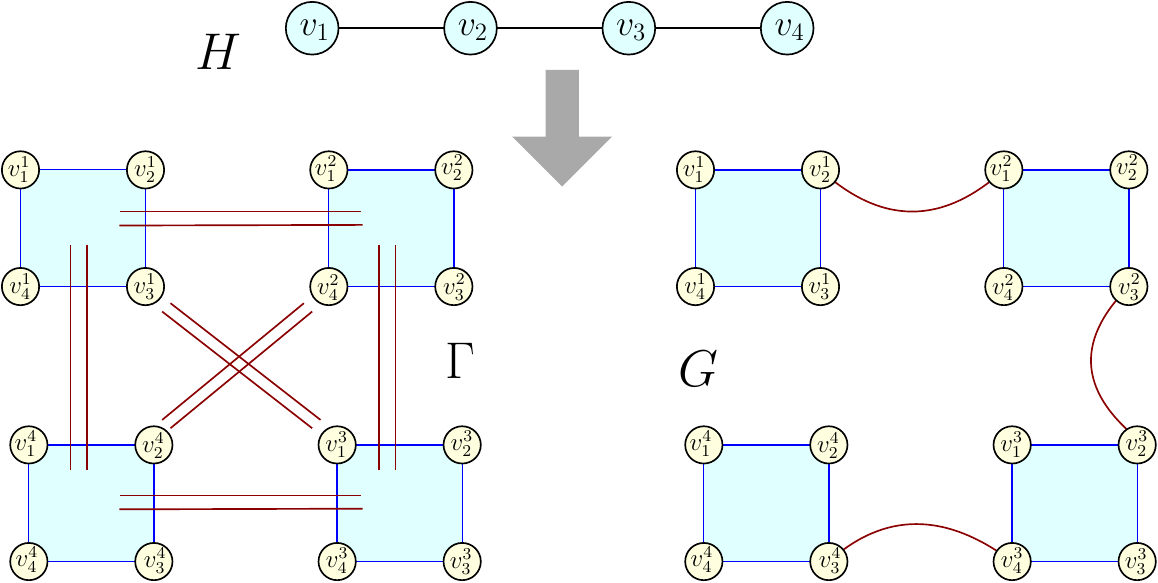}
	\caption{Hardness for $t=(2+\epsilon)$ when $\Gamma$ is an weighted graph and $G$ is a subcubic graph.}
 \label{fig:wighted}
\end{figure}

This completes the description of the instance $(G,\Gamma,k)$ with $ t=  \frac{3n}{2 \epsilon}$.  The number of vertices in $ \Gamma $ is $ n^2$,  and the construction can be done in polynomial-time. Now we give the correctness of our reduction. First, we show that the graph $G$ is a subcubic graph.

\begin{claim}\label{claim:deg3}
	The degree of each vertex in $ G $ is bounded by $3$.
	\end{claim}

\begin{proof}
	Let $ u = v_i^j $ be a vertex in $ G $ where $ i, j \in [n] $. As per construction it must have two neighbours in $ G $ which are  $ v_{i-1}^j $ and  $ v_{i+1}^j $(or $ v_1^j $). As  $ u $ is in incident to at most one edge in $ E_3 $, hence the claim.
	\end{proof}

Towards the correctness of our reduction, we prove the following. 

\begin{lemma}\label{lem:reduction}
$ (H,k) $ is a \yes-instance of {\sc Diameter-$2$ Augmentation} if and only if $(G,\Gamma,k)$ is a \yes-instance of \dla.
\end{lemma}

\begin{proof}
	In the forward direction, suppose that there exists a set $S \subseteq E(H)$ of size at most $ k $ such that the diameter of the graph $ H+S $ becomes two.  Let $ S' $ be the set of edges in $ \Gamma $ defined by $ S'= \{(v_i^j, v_j^i)~|~(v_i,v_j) \in S \}$. We show that for each pair of vertices $ x, y$ in the graph  $G  $ we have $d_{G+S'}(x,y)\leq (2+\epsilon)\cdot d_{\Gamma}(x,y)$. If $x= v_i^j$ and $ y=v_{i'}^j $ for some $ i,j,j' \in [n]$ then we 
	we already have $d_{G} (v_i^j, v_{i'}^j)=  d_{\gamma} (v_i^j, v_{i'}^j) $. So they are not at all in conflict in $ G $. Now if $x= v_i^j$ and $ y=v_{i'}^{j'} $ for some $ i,j,i',j' \in [n]$ where $ j \neq j' $ we have $ d_{\Gamma} (x,y) =w $ and $ d_{G+S'}(x,y) \leq 2w+ \frac{3n}{2}= 2w+\epsilon w= (2+\epsilon)\cdot w= (2+\epsilon) \cdot d_{\Gamma} (x,y) $.

	In the backward direction, let $ T $ be set of at most $ k $ edges in $ \Gamma $ such that for each pair of vertices $ x, y$ in the graph  $G  $ we have $d_{G+T}(x,y)\leq (2+\epsilon)\cdot d_{\Gamma}(x,y)$. Let $ T'\subseteq T $ denotes the set of those edges in $ T $ which are of the form $(v_i^j, v_{i'}^{j'}) $ for some $ i,j,i',j' \in [n]$ where $ j \neq j' $. We now define a set $ T'' $ of edges in $ H $ by  $ T''= \{(v_j, v_{j'})~|~ (v_i^j, v_{i'}^{j'}) \in T'  \} $. Now we show that the diameter of the graph $ H+T'' $ becomes two. Consider an arbitrary pair of distinct vertices  $ v_i $ and $ v_j $ in $ H $. If $ d_{H+T''}(v_i, v_j) = 2$ then we are done. Else $ d_{H+T''}(v_i, v_j) \geq  3$. That means there is no $ v_{i'} $ such that both the edges $ (v_i, v_{i'}) $ and $ (v_{i'}, v_j) $ are in $ H+T'' $. So by our construction of $ T'' $, there is no 	$ z \in [n] $ such that we have both the edges $ (v^i_a, v^z_b) $ and $ (v^z_p, v^j_q) $ are in $ E(G) \cup T' $ where $ a,b,p,q \in [n] $. That means $d_{G+T}(v^i_a,v^j_q) > 3w$ which is a contradiction due to the fact that $d_{G+T}(v^i_a,v^j_q) \leq  (2+\epsilon)\cdot w$ and $ \epsilon \in (0,1) $. \end{proof}
	
	 \cref{claim:deg3} and \cref{lem:reduction} together complete the proof of Theorem \ref{theo:weighted}. \end{proof}

  \section{Hardness of \dta in Arbitrary Metrics}

 \subsection{{\sf NP}-hardness of \dtai{2} when $G$ is Edgeless} \label{sec:genempty}
 

We prove that the \textsc{Dilation $2$-Augmentation} problem is \nph. when $G$ is an empty graph(edgeless).   Toward this we give a reduction from the {\sc $t$–Spanner} problem. Given a connected undirected unweighted graph $H$ and  a positive integers $k$, the {\sc $t$–Spanner} problem asks  to find a subset $E' \subseteq E(H)$ of size at most $k$ such that for each pair of vertices $u,v\in V(H)$ we have  $d_{H'}(u,v)\leq t\cdot d_{H}(u,v)$ where $H' = (V(H), E')$. The  {\sc $t$–Spanner} problem is known to be \nph even for $t=2$~\cite{PelegS89}.


Given an instance $(H,k)$ of {\sc $2$–Spanner} problem we construct an instance $(G,\Gamma, k')$ of \textsc{Dilation $2$-Augmentation} as follows. We set the graph $H$ as $\Gamma$. The graph $G$ is a subgraph of $\Gamma$ with $V(G)=V(\Gamma)$  and $E(G)=\emptyset$. This completes the description of the instance $(G,\Gamma,k')$ with $ k'=k$ (see Figure \ref{fig:edgeless} for an illustration of the construction).   It is easy to observe that this  construction can be done in time polynomial 
 in $V(H)$. Now we give the correctness of our reduction.

\begin{figure}[ht!]
	\centering
	\includegraphics[width=.8\textwidth]{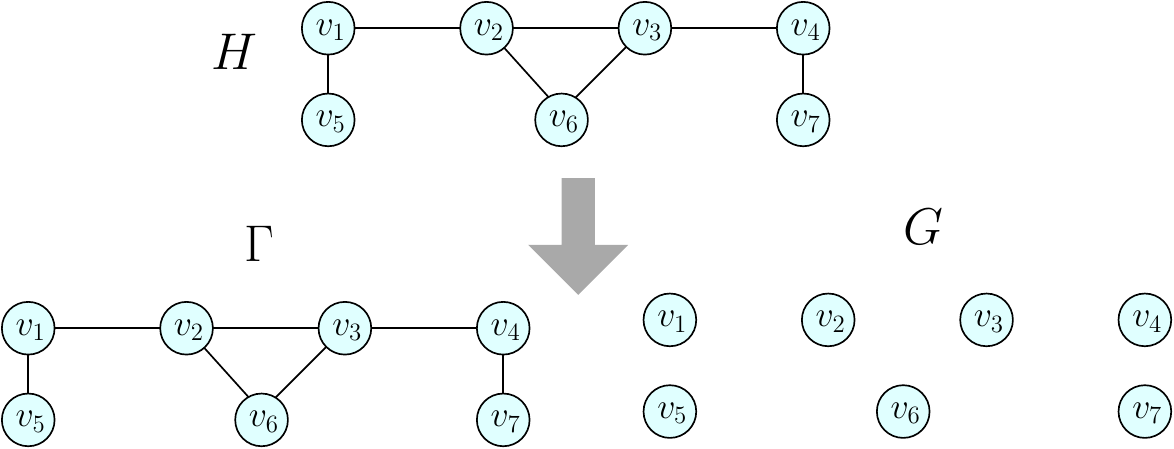}
	\caption{Hardness for $t=2$ when $G$ is an edgeless  graph.}
 \label{fig:edgeless}
\end{figure}



\begin{lemma}
  $(H,k)$ is a \yes-instance of {\sc $2$–Spanner} if and only if  $(G,\Gamma, k')$ is a \yes-instance of  {\sc Dilation $2$-Augmentation}.
\end{lemma}

\begin{proof}
In the forward direction, let $(H,k)$ be a \yes-instance of \textsc{$2$–Spanner} problem and $S \subseteq E(H)$ be a solution. That means  for every pair of vertices in the graph $H'=(V(H), S)$ we have $d_{H'}(u,v)\leq 2\cdot d_{H}(u,v)$. Consider $S=T$. Now the graph $G+T$ is precisely $H'$. By definition, for each pair of vertices in $G$ we have $d_{G+T}(u,v)= d_{H'}(u,v) \leq 2\cdot d_{H}(u,v)= 2\cdot d_{\Gamma}(u,v)$. As $|T|=|S \leq k =k'$, $T$ is a solution for  \textsc{Dilation $2$-Augmentation} in the instance $(G,\Gamma, k')$.

In the backward direction, let $(G,\Gamma, k')$ be a \yes-instance of \textsc{Dilation $2$-Augmentation} problem and $T$ be a solution.  That means for every pair of vertices $u, v$ in $G$ we have $d_{G+T} (u,v) \leq 2 \cdot d_{\Gamma}(u,v)$. Further we assume that $T$ is minimal. Now we are in two cases.

\begin{description}
     \item[{Case (i)}: $T\subseteq E(\Gamma)$.] In this case,  $d_{H'}(u,v)\leq 2 \cdot  d_{\Gamma}(u,v)= 2 \cdot  d_{\Gamma}(u,v)$, where $H'= (V(H), T)$. So $T$ is a solution for the  \textsc{$2$–Spanner} problem on $(H,k)$.
     \medskip 
\item[{Case (ii)}: $T \setminus E(\Gamma) \neq \emptyset$.]
Let $e=(p,q)$ an edge in $T \setminus E(\Gamma)$. Consider any two vertices $u$ and $v$ that are adjacent in $\Gamma$. Observe that $d_{\Gamma}(u,v)=1$ so $d_{G+T}(u,v)\leq 2$. Now since $(p,q)\notin E(\Gamma)$ we have $d_{\Gamma}(p,q)\geq 2$. So  $d_{G+T}(p,q)\geq 2$. Thus between the vertices $p$ and $q$, there exists no shortest path of length at most two containing the edge $(p,q)$ in $G+T$. This contradicts the minimality of $T$. So this case can not appear.     
\end{description}
\end{proof}

Hence we have the following proposition.

\begin{proposition}\label{prop:empty}
    {\sc Dilation $2$-Augmentation} is \nph~when $G$ is edgeless.
\end{proposition}

\subsection{{\sf W[2]}-hardness of \dtai{2} when  $\Gamma$ is  Clique} \label{sec:cliquegen}

Let $(G,\Gamma, k)$ be an instance of \dtwoa where $\Gamma$ is a complete graph, and $G$ is an arbitrary graph. Then, note that $d_{\Gamma}(u, v) = 1$ for all $u, v \in V(\Gamma)$, which implies that the weights of all the edges in $G$ is $1$, and the weight of any newly added edge is also $1$. In this case, \dtwoa is equivalent to deciding whether we can add $k$ new edges $S$ to $G$ such that the diameter of $G+S$ (i.e., maximum distance between any pair of vertices in $G+S$) is at most $2$. This corresponds to the well-known  {\sc Diameter Augmentation} problem which is known to be \wth~\cite{DBLP:journals/dam/GaoHN13}. Hence, we have the following proposition.

\begin{proposition}\label{prop:diamaug}
    \dtwoa is \wth~parameterized by $k$,  when $\Gamma$ is  clique.
\end{proposition}

\section{Conclusion} \label{Sec:conclusion}
In this article, we studied \dta and explored the multivariate complexity of the problem by restricting either the graph class to which $\Gamma$ could belong or the graph class to which $G$ could belong. Other special graph classes that one could consider include geometric intersection graphs such as interval graphs, unit-disk graphs, disk-graphs, or, more generally, string graphs. In an alternate direction, we can consider the problem in its full generality. However, as we saw, the problem in its full generality cannot admit \fpt algorithm. So, a next natural question that stems from our work is exploring these problems from the perspective of \fpt-approximation.

    \bibliography{main}

\end{document}